\documentclass[letterpaper,twocolumn,10pt]{article}
\usepackage{usenix,epsfig,endnotes}
\usepackage{times}
\usepackage{color,subfigure, float}
\usepackage[hyphens]{url}  
\usepackage{array}
\usepackage{multirow}
\usepackage{xspace}
\usepackage{verbatim}
\usepackage{wrapfig}
\usepackage[english]{babel}
\usepackage{amsmath, amsthm}
\usepackage{caption}
\usepackage[vlined,ruled]{algorithm2e}
\usepackage[shortcuts]{extdash}  
\usepackage[numbers,sort&compress]{natbib}
\usepackage{balance}

\usepackage[breaklinks=true]{hyperref}
\usepackage{breakcites}
\hypersetup{pdftex,colorlinks=true,allcolors=blue}
\usepackage{hypcap}
\usepackage{xr}

\usepackage{epstopdf}

\usepackage{breakurl}
\expandafter\def\expandafter\UrlBreaks\expandafter{\UrlBreaks\do-\do_}

\usepackage[labelfont=bf]{caption}
\DeclareMathOperator*{\argmin}{\arg\!\min}

\newcommand{\projecttitle}{\textsc{PrivApprox}\xspace}
\newcommand{\privapprox}{\textsc{PrivApprox}\xspace}
\newcommand{\myfontsize}{\fontsize{8}{9}\selectfont}
\newcommand{\commentfontsize}{\fontsize{7}{8}\selectfont}
\newcommand{\referencefontsize}{\fontsize{9}{10}\selectfont}

\newcommand*\from{\colon}

\newcommand{\myparagraph}[1]{\smallskip \noindent{\bf {#1}.}}

\newcommand{\out}[1] {}

\newcounter{codeLineCntr}

\reversemarginpar
\newif\ifnotes
\notestrue

\newcommand{\punt}[1]{}

\renewcommand{\eqref}[1]{Equation~(\ref{eq:#1})}

\newcommand{\abs}[1]{\left| #1\right|}

\newcommand{\proc}[1]{\ifmmode\mbox{\textsc{#1}}\else\textsc{#1}\fi}

  \newcommand{\func}[1]{\ifmmode\mathrm{#1}\else\textrm{#1}fi} %

\newcounter{remark}[section]

\newtheorem{theorem}{Theorem}[section]
\newtheorem{lemma}[theorem]{Lemma}
\newtheorem{proposition}[theorem]{Proposition}
\newtheorem{corollary}[theorem]{Corollary}

\begin{document}
\title{Privacy Preserving Stream Analytics\\{\large The Marriage of Randomized Response and Approximate Computing}}
\author{%
	{\em Do Le Quoc$^\dag$, Martin Beck$^\dag$, Pramod Bhatotia$^*$}\\
	{\em Ruichuan Chen$^\ddag$, Christof Fetzer$^\dag$, and Thorsten Strufe$^\dag$}\\
	{\small $^\dag$TU Dresden \quad $^*$The University of Edinburgh \quad   $^\ddag$Nokia Bell Labs}\\
	{\small Technical Report, Jan 2017}
}

\date{}
\maketitle

\maketitle

\begin{abstract}
{\em How to preserve users' \underline{privacy} while supporting \underline{high-utility} analytics for  \underline{low-latency} stream processing?}

To answer this question: we describe the design, implementation and evaluation of \projecttitle, a data analytics system for privacy-preserving stream processing. \privapprox provides three properties:  {\em (i)}   \underline{Privacy}:  zero-knowledge privacy guarantee for users, a privacy bound tighter than the state-of-the-art differential privacy; {\em (ii)}   \underline{Utility}: an interface for data analysts to
systematically explore the trade-offs  between the output accuracy (with error-estimation) and the query execution budget; {\em (iii)}   \underline{Latency}: near real-time stream processing based on a scalable ``synchronization-free''  distributed architecture.

The key idea behind our approach is to marry two techniques together, namely, {\em sampling} (used in the context of approximate computing) and {\em randomized response} (used in the context of privacy-preserving analytics). The resulting marriage is complementary --- it achieves stronger privacy guarantees and also improves the performance for low-latency stream analytics.
\end{abstract}
 \section{Introduction}
\label{sec:introduction}

Many online services continuously collect users' private data for real-time analytics. Much of this data arrives as a data stream and in huge volumes, requiring real-time stream processing based on distributed systems~\cite{flink, s4, storm, spark-streaming}.

In the current ecosystem of data analytics, the analysts usually have direct access to the users' private data, and must be trusted not to abuse it. However, this trust has been violated in the past~\cite{violation1,violation2,violation3,violation4}.

A pragmatic eco-system has two desirable, but contradictory design requirements:  {\em (i)} stronger privacy guarantees for the users; and {\em (ii)} high-utility stream analytics in real-time.  Users seek stronger privacy, while analysts strive for high-utility analytics in real time.

To meet these two design requirements,  there is a surge of novel computing paradigms that address these concerns, albeit {\em separately}. Two such paradigms are {\em privacy-preserving analytics} to protect user privacy and {\em approximate computation} for real-time analytics.

\myparagraph{Privacy-preserving analytics} Recent privacy-preserving analytics systems favor a distributed architecture to avoid central trust (see $\S$\ref{sec:related} for details), where users' private data is stored locally on their respective client devices. Data analysts use a publish-subscribe mechanism to run aggregate queries over the distributed private dataset of a large number of clients. Thereafter, such systems add noise to the aggregate output to provide useful privacy guarantees, such as differential privacy~\cite{differential-privacy}. Unfortunately, these state-of-the-art systems normally deal with single-shot batch queries, and therefore, these systems cannot be used for real-time stream analytics.

\myparagraph{Approximate computation}  Approximate computation is based on the observation that many data analytics jobs are amenable to an approximate, rather than the exact output (see $\S$\ref{sec:related} for details). For such an approximate workflow, it is possible to trade accuracy by  computing over a partial subset (usually selected via a sampling mechanism) instead of the entire input dataset. Thereby, data analytics systems based on approximate computation can achieve low latency and efficient utilization of resources.  However, the existing systems for approximate computation assume a centralized dataset, where the desired sampling mechanism can be employed. Thus, existing systems are not compatible with the distributed privacy-preserving analytics systems.

\myparagraph{The marriage} In this paper, we make the observation that the two computing paradigms, privacy-preserving analytics and approximate computation, are complementary. Both  paradigms strive for an approximate instead of the exact output, but they differ in their {\em means} and {\em goals} for approximation. Privacy-preserving analytics adds explicit {\em noise} to the aggregate query result to protect users' privacy. Whereas, approximate computation relies on a representative {\em sampling}  of the entire dataset to compute over only a subset of data items to enable low-latency/efficient analytics. Therefore, we marry these two existing paradigms together in order to leverage the benefits of both. The high-level idea is to achieve privacy (via approximation) by directly computing over a subset of sampled data items (instead of computing over the entire dataset) and then adding an explicit noise for privacy-preservation.

To realize this marriage, we designed an approximation mechanism that also achieves privacy-preserving goals for stream analytics. Our design (see Figure~\ref{fig:system-overview}) targets a distributed setting, similar as aforementioned, where users' private data is stored locally on their respective personal devices, and an analyst issues a streaming query for analytics  over the distributed private dataset of users. The analyst's streaming query is executed on the users' data periodically (a configurable epoch) and the query results are transmitted to a centralized aggregator via a set of proxies. The analyst interfaces with the aggregator to get the aggregate query output periodically.

We employ two core techniques to achieve our goal. Firstly, we employ {\em sampling}~\cite{srs-sampling} directly at user's site for approximate computation, where each user randomly decides whether to participate in answering the query in the current epoch. Since we employ sampling at the data source, instead of sampling at a centralized infrastructure, we are able to squeeze out the desired data size (by controlling the sampling parameter) from the very ``first stage" in the analytics pipeline, which is  essential in low-latency environments.

Secondly, if the user participates in the query answering process, we employ a {\em randomized response}~\cite{fox1986randomized} mechanism to add noise to the query output at user's site, again locally at the source of the data in a decentralized fashion. In particular, each user locally randomizes its truthful answer to the query to achieve local differential privacy guarantees ($\S$\ref{subsec:answer:randomization}). 
Since we employ noise addition at the source of data, instead of adding the explicit
noise to the aggregate output at a trusted aggregator or proxies, we enable a truly ``synchronization-free" distributed architecture, which requires {\em no coordination} among proxies and the aggregator for  the mandated noise addition.

The last, but not the least, silver bullet of our design: it turns out that the combination of the two aforementioned techniques (i.e., sampling and randomized response) led us to achieve zero-knowledge privacy~\cite{zero-knowledge-privacy}, a privacy bound tighter than the state-of-the-art differential privacy~\cite{differential-privacy}.  (We prove our claim in Appendix~\ref{sec:privacy-evaluation}.)

To summarize, we present the design and implementation of a practical system for privacy-preserving stream analytics in real time. In particular,  our system is a novel combination of the sampling and randomized response techniques, as well as a scalable ``synchronization-free" routing scheme employing a light-weight XOR encryption scheme~\cite{Chen-2013}.  The resulting system ensures zero-knowledge privacy, anonymization, and unlinkability for users ($\S$\ref{subsec:privacy-properties}). 
Altogether, we make the following contributions:

\begin{itemize}

\item We present a marriage of sampling and randomized response to achieve improved performance and stronger privacy guarantees.

\item We present an adaptive query execution interface for analysts to systematically make a trade-off between the output accuracy, and the query execution budget.

\item We present a confidence metric on the output accuracy using a confidence interval to interpret the approximation due to sampling and randomization.

\end{itemize}

To empirically evaluate our approach, we implemented our design as a fully-functional prototype in a system called \projecttitle\footnote{\label{note1} The source code of \projecttitle along with the experimental evaluation setup is publicly available :\href{https://privapprox.github.io}{https://PrivApprox.github.io}.}  based on Apache Flink~\cite{flink} and Apache Kafka~\cite{kafka}.  In addition to stream analytics, we further extended our system to support privacy-preserving ``historical" batch analytics over users' private datasets. The evaluation  based on micro-benchmarks and real-world case-studies shows that this marriage is, in fact, made in heaven!

\section{Overview}
\label{sec:overview}
\begin{figure}[t]
\centering
\includegraphics[scale=0.33]{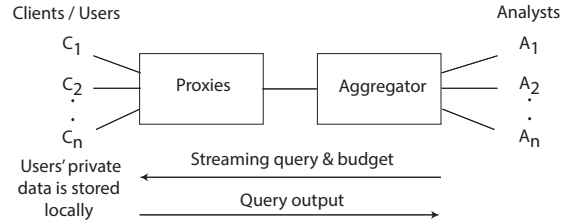}
\caption{System overview}
\label{fig:system-overview}
\end{figure}

In this section, we present an overview of our system called \projecttitle. 

\subsection{System Architecture}
\label{subsec:overview}

\projecttitle is designed for privacy-preserving stream analytics on distributed users' private dataset. Figure \ref{fig:system-overview} depicts the high-level architecture of \projecttitle. Our system consists of four main components: clients, proxies, aggregator, and analysts.

{\em Clients} locally store users' private data on their respective personal devices, and subscribe to queries from the system.  {\em Analysts} publish streaming queries to the system, and also specify a query execution budget.  The query execution budget can either be in the form of latency guarantees/SLAs, output quality/accuracy, or the available computing resources for query processing.  Our system ensures that the computation remains within the specified budget.

At a high-level, the system works as follows: a query published by an analyst is distributed to clients via the aggregator and proxies. Clients answer the analyst's query locally over the users'  private data using a privacy-preserving mechanism. Client answers are transmitted to the aggregator via anonymizing {\em proxies}. The {\em aggregator} aggregates received answers from the clients to provide privacy-preserving stream analytics to the analyst.

\subsection{System Model}
\label{subsec:system-model}
Before we explain the design of \projecttitle, we present the system model assumed in this work.

\subsubsection{Query Model}
\label{subsubsec:query-model}
\projecttitle supports the SQL query language for {\em analysts} to formulate streaming queries. While queries can be complex, the results of a query are expressed as counts within histogram buckets, i.e., each bucket represents a small range of query's answer values. Specifically, each query answer is represented in the form of binary buckets, where each bucket stores a possible answer value `1' or `0' depending on whether or not the answer falls into the value range represented by that bucket.
For example, an analyst can learn the driving speed distribution across all vehicles in San Francisco by formulating an SQL query ``\texttt{SELECT speed FROM vehicle WHERE location=`San Francisco'}''.  The analyst can then define 12 answer buckets on speed: `0', `1$\sim$10', `11$\sim$20', $\cdots$, `81$\sim$90', `91$\sim$100', and `$>100$'.  If a vehicle is moving at 15 mph in San Francisco, it answers `1' for the third bucket and `0' for all others.

Our query model supports not only numeric queries as described above, but also non-numeric queries.  For non-numeric queries, each bucket is specified by a matching rule or a regular expression.  
Note that, at first glance, our query model may appear simple, it however supports  a range of queries such as histogram queries and frequency queries. In addition, it has been shown to be effective for a wide-range of analytics algorithms~\cite{BlumDMN05,statistical-analysis}.

\subsubsection{Computation Model}  \projecttitle adopts a {\em batched stream} programming model~\cite{flink, spark-streaming} in which the online data stream is split into small batches; and each small batch is processed by launching a distributed data-parallel job. The batched streaming model is adopted widely compared to trigger-based systems~\cite{s4, storm} for the following advantages: exact-once semantics, efficient fault-tolerance, and a common data-parallel programming model for both stream and batch analytics.

In particular, \projecttitle employs {\em sliding window computations} over batched stream processing~\cite{slider, contraction-tree}. For sliding windows,  the computation window slides over the input data stream, where the new incoming data items are added, and the old data items are dropped from the window as they become less relevant. Note that these systems~\cite{flink, spark-streaming} expose a time-based window length, and based on the arrival rate,  the number of data items within a window may vary accordingly.

\subsubsection{Threat Model}
\label{subsubsec:threat-model}

{\em Analysts} are potentially malicious.  They may try to violate the \projecttitle's privacy model, i.e., de-anonymize clients, build profiles through the linkage of requests and answers, or de-rand\-omize (remove added noise from) the answers.

{\em Clients} are potentially malicious. They could generate false or invalid responses to distort the query result for the analyst. However, we do not defend against the Sybil attack~\cite{sybil}, which is beyond the scope of this work~\cite{sybil-mobisys}. 

{\em Proxies} are also potentially malicious. They may transmit messages between clients and the aggregator in contravention of the system protocols. 
 \projecttitle includes at least two proxies, and there are at least two proxies which do not collude with each other. 

The {\em aggregator} is assumed to be Honest-but-Curious (HbC): the aggregator faithfully conforms to the system protocol, but may try to exploit the information about clients. The aggregator does not collude with any proxy, nor the analyst. 

Finally, we assume that all end-to-end communications use authenticated and confidential connections (are protected by long-lived TLS connections), and no system component could monitor all network traffic.

\subsubsection{Privacy Properties}
\label{subsec:privacy-properties}
Our privacy properties include:  {\em (i)}~zero-knowledge privacy,  {\em (ii)}~anonymity, and {\em (iii)}~unlinkability.

All aggregate query results in the system are independently produced under {\em zero-knowledge privacy} guarantees. The chosen privacy metric {\em zero-knowledge privacy}~\cite{zero-knowledge-privacy} builds upon differential privacy~\cite{differential-privacy} and provides a tighter bound on privacy guarantees compared to differential privacy. Informally, zero-knowledge privacy states that essentially everything that an adversary can learn from the output of an zero-knowledge private mechanism could also be learned using aggregate information.
{\em Anonymity} means that no system components can associate query answers or query requests with a specific client.  Finally,  {\em unlinkability} means that no system component can join any pair of query requests or answers to the same client, even to the same anonymous client.

For the formal definitions, analysis, and proofs---refer Appendix~\ref{sec:privacy-evaluation}.

\subsubsection{Assumptions} 
\label{sec:design-assumptions}
We make the following assumptions.
 
\begin{itemize}

\item[1.] We assume that the input stream is stratified based on the source of event, i.e., the data items within each stratum follow the same distribution, and are mutually independent. Here a \textit{stratum} refers to one sub-stream. If multiple sub-streams have the same distribution, they are combined to form a stratum. 

\item[2.] We assume the existence of a virtual function that takes the query budget as the input and outputs the sample size for each window based on the budget.

\item[3.] We assume that the  aggregator faithfully follows the system protocol. We could use  trusted computing such as remote attestation~\cite{excalibur} based on Trusted Platform Modules (TPMs) to relax the HbC assumption.

\end{itemize}

 We discuss different possible means to meet the first two assumptions in Appendix~\ref{sec:discussion}.
\section{Design}
\label{sec:design}

\projecttitle consists of two main phases (see Figure~\ref{fig:system-overview}): \emph{submitting queries} and \emph{answering queries}. In the first phase, an analyst submits a query (along with the execution budget)  to clients via the aggregator and proxies. In the second phase, the query is answered by the clients in the reverse direction.

\subsection{Submitting Queries}
\label{subsec:submitting-queries}

 To perform statistical analysis over users' private data streams, an analyst creates a query using the query model described in $\S$\ref{subsubsec:query-model}.  In particular, each query consists of the following fields, and is signed by the analyst for non-repudiation:
 
\begin{equation}
\label{eq:query}
 Query := \langle Q_{ID}, SQL, A[n], f, w, \delta \rangle
\end{equation}

\begin{itemize}

\item $Q_{ID}$ denotes a unique identifier of the query.  This can be generated by concatenating the
identifier of the analyst with a serial number unique to the analyst.

\item $SQL$ denotes the actual SQL query, which is passed on to clients and executed on their respective personal data.

\item $A[n]$ denotes the format of a client's answer to the query. The answer is an $n$-bit vector where each bit associates with a possible answer value in the form of a ``0'' or ``1'' per index (or answer value range).  

\item $f$ denotes the answer frequency, i.e., how often the query needs to be executed at clients.

\item $w$ denotes the window length for sliding window computations~\cite{slider}. For example, an analyst may only want to aggregate query results for the last ten minutes, which means the window length is ten minutes.

\item $\delta$ denotes the sliding interval for sliding window computations. For example, an analyst may want to update the query results every one minute, and so the sliding interval is set to one minute.
\end{itemize}

After forming the query, the analyst sends the query, along with the query execution budget, to the aggregator. Once  receiving the pair of the query and query budget from the analyst, the aggregator first converts the query budget into system parameters for sampling ($s$) and randomization ($p, q$). We explain these system parameters in the next section $\S$\ref{subsec:answering}. Hereafter, the aggregator forwards the query and the converted system parameters to clients via proxies.

\subsection{Answering Queries}
\label{subsec:answering}

After receiving the query and system parameters, we next explain how the query is answered by clients and processed by the system to produce the result for the analyst. 
The query answering process involves several steps including {\em (i)} sampling at clients for low-latency approximation; {\em (ii)}  randomizing answers for privacy preservation; {\em (iii)}  transmitting answers for anonymization and unlinkability; and  finally, {\em (iv)}  aggregating answers with error estimation to give a confidence level on the approximate output. We next explain the entire workflow using these four steps. (The algorithms are detailed in Appendix~\ref{sec:algorithms}.)

\subsubsection{Step I: Sampling at Clients}
\label{subsec:answer:sampling}

We make use of  approximate computation to achieve low-latency execution by computing over a subset of data items instead of the entire input dataset. Specifically, our work builds on sampling-based techniques~\cite{BlinkDB, quickr-sigmod, approxhadoop, incapprox-www-2016} in the context of ``Big Data'' analytics.
Since we aim to keep the private data stored at individual clients, \projecttitle applies an input data sampling mechanism locally at the clients. In particular, we use {\em Simple Random Sampling}
(SRS)~\cite{srs-sampling}.

\myparagraph{Simple Random Sampling (SRS)}  SRS is considered as a fair way of selecting a sample from a given population since each individual in the population has the same chance of being included in the sample. We make use of SRS at the clients to select clients that will participate in the query answering process. In particular, the aggregator passes the {\em sampling parameter} ($s$) on to clients as the probability of participating in the query answering process. Thereafter,  each client flips a coin with the probability based on the sampling parameter ($s$), and decides whether to participate in answering a query.  Suppose that we have a population of $U$ clients, and each client $i$ has an answer $a_i$. We want to calculate the sum of these answers across the population, i.e., $\sum_{i=1}^{U} a_i $. To compute an approximate sum, we apply the SRS at clients to get a sample of $U'$ clients. The estimated sum is then calculated as follows:
\begin{equation}
\label{eq:srs-estimated-sum}
\hat{\tau} = \dfrac{U}{U'}\sum\limits_{i=1}^{U'} a_i  \pm error
\end{equation}

\noindent Where the error bound  $error$ is defined as:
\begin{equation}
\label{eq:error-bound}
error = t\sqrt{\widehat{Var}(\hat{\tau})}
\end{equation}

\noindent Here, $t$ is a value of the $t$-distribution with $U' - 1$ degrees of freedom at the $1 - \alpha/2$ level of significance, and the estimated variance $\widehat{Var}(\hat{\tau})$ of the sum is:
\begin{equation}\label{eq:variance_sum}
\widehat{Var}(\hat{\tau}) = \dfrac{U^2}{U'}\sigma^2(\dfrac{U -  U'}{U})
\end{equation}

\noindent Where $\sigma^2$ is the sample variance of sum.

Note that we currently assume that all clients produce the input stream with data items following the same distribution, i.e., all clients' data streams belong to the same stratum.  We further extend it for stratified sampling in~$\S$\ref{sec:extensions}.

\subsubsection{Step II: Answering Queries at Clients}
\label{subsec:answer:randomization}

Clients that participate in the query answering process make use of the {\em randomized response} technique~\cite{fox1986randomized} to preserve answer privacy, with \emph{no} synchronization among clients.

\myparagraph{Randomized response} Randomized response protects user's privacy by allowing individuals to answer sensitive queries without providing truthful answers all the time, yet it allows analysts to collect statistical results. Randomized response works as follows: suppose an analyst sends a query to individuals to obtain the statistical result about a sensitive property. To answer the query, a client locally randomizes its answer to the query~\cite{fox1986randomized}.  Specifically, the client flips a coin, if it comes up heads, then the client responds its truthful answer; otherwise, the client flips a second coin and responds ``Yes" if it comes up heads or ``No" if it comes up tails. The privacy is preserved via the ability to refuse responding truthful answers.

Suppose that the probabilities of the first coin and the second coin coming up heads are $p$ and $q$, respectively. The analyst receives $N$ randomized answers from individuals, among which $R_y$ answers are ``Yes". Then, the number of original truthful ``Yes" answers before the randomization process can be estimated as:
\begin{equation}
\label{eq:estimatedAnswer}
E_y = \dfrac{R_y - (1- p)\times q \times N}{p}
\end{equation}

Suppose $A_y$ and $E_y$  are the actual and the estimated numbers of the original truthful ``Yes" answers, respectively. The accuracy loss $\eta$ is then defined as:
\begin{equation}
\label{eq:accuracyloss}
\eta = \abs{ \dfrac{A_y - E_y}{A_y} }
\end{equation}

It has been proven in~\cite{DBLP:journals/fttcs/DworkR14} that, the randomized response mechanism achieves $\epsilon$-differential privacy~\cite{differential-privacy}, where:
\begin{equation}
\label{eq:privacy-level1}
\varepsilon = \ln\big(\dfrac{\Pr[Response = Yes | Truth = Yes]}{\Pr[Response = Yes | Truth = No]} \big)
\end{equation}

More specifically, the randomized response
mechanism achieves $\epsilon$-differential privacy, where:
\begin{equation}
\label{eq:privacy-level2}
\varepsilon = \ln\big(\dfrac{p + (1 - p) \times q}{(1 - p) \times q}\big)
\end{equation}

The reason is: if a truthful answer is ``Yes", then with the probability of `$p + (1 - p) \times q$', the randomized answer will still remain ``Yes". Otherwise, if a truthful answer is ``No", then with the probability of `$(1 - p) \times q$', the randomized answer will become ``Yes".

It is worth mentioning that, combining randomized response with the sampling technique used in Step I, we achieve not only differential privacy but also zero-knowledge privacy~\cite{zero-knowledge-privacy} which is a privacy bound tighter than differential privacy. We prove our claim in Appendix~\ref{sec:privacy-evaluation}.

\begin{figure}[t]
\centering
\includegraphics[scale=0.35]{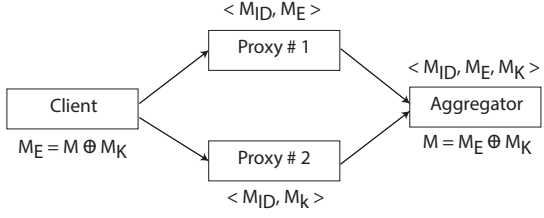}
\caption{XOR-based encryption with two proxies.}
\label{fig:queryans}
\end{figure}

\subsubsection{Step III: Transmitting Answers via Proxies}
\label{subsec:answer:otp}

After producing randomized responses, clients transmit them to the aggregator via the proxies. To achieve anonymity and unlinkability of the clients against the aggregator and analysts, we utilize the XOR-based encryption together with source rewriting, which has been used for anonymous communications~\cite{Chen-2012, Chen-2013,onion-routing,tor}. 
Under the assumptions that:
\begin{itemize}
  \item at least two proxies are not colluding
  \item the proxies don't collude with the aggregator, nor the analyst
  \item the aggregator and analyst have only a local view of the network
\end{itemize}
neither the aggregator, nor the analyst will learn any (pseudo-)identifier to deanonymize or link different answers to the same client. This property is achieved by source rewriting, which is a typical building block for anonymization schemes~\cite{onion-routing,tor}. At the same time the content of the answers is hidden from the proxies using the XOR-based encryption.

\myparagraph{XOR-based encryption} At a high-level, the XOR-based encryption employs extremely efficient bit-wise XOR operations as its cryptographic primitive compared to expensive public-key cryptography. This allows us to support resource-constrained clients, e.g., smartphones and sensors. The underlying idea of this encryption is simple:  if Alice wants to send a message $M$ of length $l$ to Bob, then  Alice and Bob share a secret $M_K$ (in the form of a random bit-string of length
$l$). To transmit the message $M$ privately, Alice sends an encrypted message `$M_E = M \oplus M_K$' to Bob, where `$\oplus$' denotes the bit-wise XOR operation. To decrypt the message,  Bob again uses the bit-wise XOR operation: $M = M_E \oplus M_K$.

Specifically, we apply the XOR-based encryption to transmit clients' randomized answers as follows. At first, each randomized answer is concatenated with its associated query identifier $Q_{ID}$ to build a message $M$:
\begin{equation}
\label{eq:message}
M = Q_{ID}, RandomizedAnswer
\end{equation}

Thereafter, the client generates $(n-1)$ random $l$-bit key strings $M_{K_i}$ with $2 \leq i \leq n$ using a cryptographic pseudo-random number generator (PRNG) seeded with a cryptographically strong random number. The XOR of all $(n-1)$ key strings together forms the secret $M_K$.
\begin{equation}
M_K = \bigoplus_{i=2}^n M_{K_i}
\end{equation}

Next, the client performs an XOR operation with $M$ and $M_K$ to produce an encrypted message $M_E$.
\begin{equation}
M_E = M \oplus M_K
\end{equation}

As a result, the message $M$ is split into $n$ messages $\langle M_E, M_{K_2}, \cdots, M_{K_n}\rangle$. Afterwards, a unique message identifier $M_{ID}$ is generated, and sent along with the split messages to the $n$ proxies via anonymous channels enabled by source rewriting~\cite{onion-routing,tor}.
\begin{equation}
\label{eq:transmit}
\begin{split}
\textnormal{Client} \longrightarrow \textnormal{Proxy}1: \langle M_{ID}, M_E \rangle\\
\textnormal{Client} \longrightarrow \textnormal{Proxy}i: \langle M_{ID}, M_{K_i} \rangle
\end{split}
\end{equation}

Upon receiving the messages (either  $\langle M_{ID}, M_E \rangle$ or $\langle M_{ID}, M_{K_i} \rangle$) from clients, the $n$ proxies transmit these messages to the aggregator. 

The message identifier $M_{ID}$ ensures that $M_E$ and all associated $M_{K_i}$ will be joined later to decrypt the original message $M$ at the aggregator. Note that, $\langle M_{ID}, M_E \rangle$ and all $\langle M_{ID}, M_{K_i} \rangle$ are computationally indistinguishable, which hides from the proxies if the received data contains the encrypted answer or is just a pseudo\-/random bit string.

\subsubsection{Step IV: Generating Result at the Aggregator}
\label{subsec:answer:aggregation}

At the aggregator, all data streams ($\langle M_{ID}, M_E \rangle$ and  $\langle M_{ID}, M_{K_i} \rangle$) are received, and can be joined together to obtain a unified data stream. Specifically, the associated $M_E$ and $M_{K_i}$ are paired by using the message identifier $M_{ID}$. To decrypt the original randomized message $M$ from the client, the XOR operation is performed over $M_E$ and $M_K$: $M = M_E \oplus M_K$ with $M_K$ being the XOR of all $M_{K_i}$: $M_K =
\bigoplus_{i=2}^n M_{K_i}$. As the aggregator cannot identify which of the received messages is $M_E$, it just XORs all the $n$ received messages to decrypt $M$.

The joined answer stream is processed to produce the query results as a sliding window.  
For each window, the aggregator first adapts the computation window to the current start time $t$ by removing all old data items, with $timestamp < t$, from the window.
Next, the aggregator adds the newly incoming data items into the window. Then,  the answers in the window are decoded and aggregated to produce the query results for the analyst.  Each query result is an estimated result which is bound to a range of error due to the approximation. The aggregator estimates this error bound using equation~\ref{eq:error-bound} and produces a confidence interval for the result as: $queryResult \pm errorBound$. The entire process is repeated for every window.

Note that an adversarial client might answer a query many times in an attempt to distort the query result. However, we can handle this problem, for example, by applying the {\em triple splitting} technique~\cite{Chen-2013}.
 
\myparagraph{Error bound estimation} 
We provide an error bound estimation for the aggregate query results. The accuracy loss in \projecttitle is caused by two processes: {\em (i)} sampling and {\em (ii)} randomized response. Since the accuracy loss of these two processes is statistically independent (see $\S$\ref{sec:evaluation}), we estimate the accuracy loss of each process separately. Furthermore, Equation~\ref{eq:srs-estimated-sum} indicates that the error induced by sampling can be described as an additive component of the estimated sum. The error induced by randomized response is contained in the \(a_i\) values in Equation~\ref{eq:srs-estimated-sum}. Therefore, independent of the error induced by randomized response, the error coming from sampling is simply being added upon. Following this, we sum up both independently estimated errors to provide the total error bound of the query results.

To estimate the accuracy loss of the randomized response process, we make use of an experimental method. We run several micro-benchmarks at the beginning of the query answering process without performing the sampling process, to estimate the accuracy loss caused by randomized response. We measure the accuracy loss using Equation~\ref{eq:accuracyloss}.

On the other hand, to estimate the accuracy loss of the sampling process, we apply the statistical theory of the sampling techniques. In particular, we first identify a desired confidence level, e.g., $95$\%. Then, we compute the margin of error  using Equation~\ref{eq:error-bound}. Note that, to use this equation the sampling distribution must be nearly normal. According to the Central Limit Theorem (CLT), when the sample size $U'$ is large enough (e.g., $\geq 30$), the sampling distribution of a statistic becomes close to the normal distribution, regardless of the underlying distribution of values in the dataset~\cite{samplingBySteve}.
\subsection{Practical Considerations}
\label{sec:extensions}
Next, we present three design enhancements to improve the practicality of \projecttitle.

\subsubsection{Stratified Sampling}
\label{subsec:stratified}

As described  in $\S$\ref{subsec:answer:sampling}, we employ Simple Random Sampling (SRS) at clients for approximate computation. The assumption behind using SRS is that all clients produce data streams following the same distribution, i.e., all clients' data streams belong to the same {\em stratum}.   However, in a distributed environment, it may happen that different clients produce data streams with disparate distributions.

Accommodating such cases requires that all strata are considered fairly to have a representative sample from each stratum. To achieve this we use the stratified sampling technique~\cite{incapprox-www-2016, BlinkDB}. Stratified sampling ensures that data from every stratum is proportionally selected (based on the arrival rate) and none of the minorities are excluded. 

To perform stratified sampling, instead of just one sampling parameter $s$, we use a set of sampling parameters $S=\{s_i\}$ where $i \in \{1, \cdots, n\}$ ($n$ is the number of disparate distribution sub-streams in the input stream). All clients within a given stratum $i$ flip a sampling coin with the probability $s_i$ to decide on the participation in the answering process. The value $s_i$ is determined based on the proportional arrival rate of the sub-stream (or stratum). The rest of the answering process remains unchanged (as in $\S$\ref{subsec:answer:randomization}).

Accordingly, we adapt the error estimation for stratified sampling to provide a confidence interval for the query result.
Suppose the clients $C$ come from $n$ sources (disjoint strata) $C_{1}$, $C_{2}$  $\cdots,$ $C_{n}$, i.e., $C = \sum_{i=1}^{n}C_{i}$, and the $i^{\textrm{th}}$ stratum $C_{i}$ has $B_{i}$ clients and each such client $j$ has an associated answer $a_{ij}$ in binary format. 

To compute an approximate sum of the ``Yes'' answers, we first select a sample from all clients $C$ based on the stratified sampling, i.e., we sample $b_{i}$ items from each $i^{\textrm{th}}$ stratum $C_{i}$. Then we estimate the sum from this sample as: $\hat{\tau} = \sum_{i=1}^{n} ( \frac{B_{i}}{b_{i}} \sum_{j=1}^{b_{i}}  a_{ij} ) \pm \epsilon$ 
where the error bound $\epsilon$ is defined as:
\begin{equation}\label{eq:epsilon}
\epsilon = t_{f,1-\frac{\alpha}{2}} \sqrt{\widehat{Var}(\hat{\tau})}
\end{equation}

Here, $t_{f,1-\frac{\alpha}{2}}$ is the value of the $t$-distribution (i.e., \textit{t-score}) with $f$ degrees of freedom and $\alpha$ $=$ $1$ $-$ confidence level. The degree of freedom $f$ is calculated as:
\begin{equation}\label{eq:degree_of_freedom}
f =  \sum_{i=1}^{n} b_{i} - n
\end{equation}

The estimated variance for the sum, $\widehat{Var}(\hat{\tau})$, can be expressed as:
\begin{equation}\label{eq:variance_sum}
\widehat{Var}(\hat{\tau}) = \sum\limits_{i=1}^n B_{i} * (B_{i} - b_{i}) \dfrac{r^2_{i}}{b_{i}}
\end{equation}

Here, $r^2_{i}$ is the population variance in the $i^{\textrm{th}}$ stratum. Similar to the SRS described in $\S$\ref{subsec:answer:sampling}, we use the statistical theories~\cite{samplingBySteve} for stratified sampling to calculate the error bound.

\subsubsection{Historical Analytics}
\label{subsec:historical}
In addition to providing real-time data analytics, we further extended \projecttitle to support historical analytics. The historical analytics workflow is essential for the data warehousing setting, where analysts wish to analyze user behaviors over a longer time period.  To facilitate historical analytics, we support the ``batch analytics" over the users' data at the aggregator. The analyst can analyze users' responses stored in a fault-tolerance distributed storage
(HDFS) at the aggregator to get the aggregate query result over the desired time period.

We further extend the adaptive execution interface for historical analytics, where the analyst can specify query execution budget, for example, to suit dynamic pricing in spot markets in the cloud deployment. Based on the query budget, we perform an additional round of sampling at the aggregator to ensure that batch analytics computation remains within the query budget. We omit the sampling details at the aggregator due to space constraints.

\subsubsection{Query Inversion}
\label{subsec:negation}

In the current setting, some queries may result in very few ``Yes'' truthful answers in users' responses. For such cases,   \privapprox can only achieve lower utility as the fraction of truthful ``Yes" answers gets far from the second randomization parameter $q$ (see experimental results in $\S$\ref{subsec:truthful-answers}). For instance, if $q$ is set to a high value (e.g., $q = 0.9$), having few ``Yes'' answers in the user responses will affect the overall utility of the query result. 

To address this issue, we propose a {\em query inversion} mechanism.  If the fraction of truthful ``Yes'' answers is too small or too large compared to the $q$ value, then the  analysts can invert the query to calculate the truthful ``No'' answers instead of the truthful ``Yes'' answers. In this way, the fraction of truthful ``No'' answers gets closer to $q$, resulting in a higher utility of the query result.
\section{Implementation}
\label{sec:implementation}

We implemented \projecttitle as an end-to-end stream analytics system.  Figure~\ref{fig:implementation} presents the architecture of our prototype.  Our system implementation consists of three main components: {\em (i)} clients,  {\em (ii)}  proxies, and {\em (iii)} the aggregator.

First, the query and the execution budget specified by the analyst are processed by the {\tt initializer} module to decide on  the sampling parameter ($s$) and the randomization parameters ($p$ and $q$). These parameters along with the query are then sent to the clients.

\myparagraph{Clients}  We implemented Java-based clients for mobile devices as well as for personal computers. 
A client makes use of the sampling parameter (based on the {\tt sampling} module) to decide whether to participate in the query answering process ($\S$\ref{subsec:answer:sampling}). If the client decides to participate then the {\tt query answer} module is used to execute the input query on the local user's private data stored in {\tt SQLite}~\cite{sqlite}. The client makes use of the randomized response to execute the query ($\S$\ref{subsec:answer:randomization}). Finally, the randomized answer is encrypted
using the {\tt XOR-based en\-cryption} module; thereafter, the encrypted message and the key messages are sent to the aggregator via proxies ($\S$\ref{subsec:answer:otp}). 

\myparagraph{Proxies} We implemented proxies based on Apache Kafka (which internally uses Apache Zookeeper~\cite{zookeeper} for fault tolerance). 
In Kafka, a {\em topic} is used to define a stream of data items. A stream {\em producer} can publish data items to a topic, and these data items are stored in Kafka servers called {\em brokers}. Thereafter, a {\em consumer} can subscribe to the topic and consume the data items by pulling them from the brokers. In particular, we make use of Kafka APIs to create two main  topics: \emph{key} and  \emph{answer} for transmitting the key message stream and the encrypted answer stream in the XOR-based encryption protocol, respectively ($\S$\ref{subsec:answer:otp}). 
\begin{figure}[t]
\centering
\includegraphics[scale=0.32]{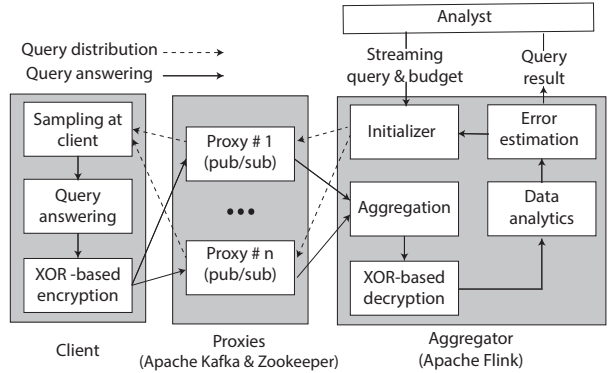}
\caption{Architecture of \projecttitle prototype}
\label{fig:implementation}
\end{figure}

\begin{figure*}[t]
\centering
\includegraphics [width=1\textwidth]{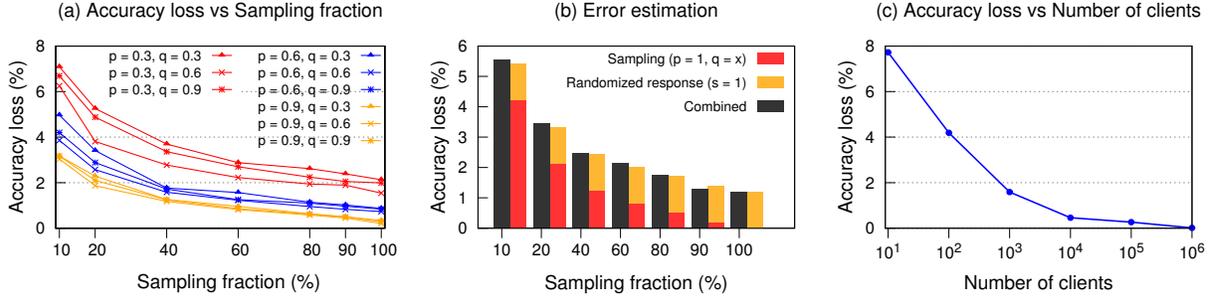}
\caption{{\bf (a)} Accuracy loss with varying sampling and randomization parameters.  
{\bf (b)} Error estimation during the randomized response process and sampling process, combined and individually. 
{\bf (c)} Accuracy loss with varying \# of clients. }

\vspace{-3mm}
\label{fig:benchmark1}
\end{figure*}

\myparagraph{Aggregator}  We implemented the aggregator using Apache Flink for real-time stream analytics and also for historical batch analytics. 
At the aggregator, we first make use of the join method (using the {\tt aggregation} module) to combine the two data streams: {\em (i)} encrypted answer stream and {\em (ii)} key stream. Thereafter, the combined message stream is decoded (using the {\tt XOR-based decryption} module) to reproduce the randomized query answers. These answers are then forwarded to the analytics module.  The {\tt analytics} module processes the answers to provide the query result to the analyst.
Moreover, the {\tt error estimation module} is used to estimate the error ($\S$\ref{subsec:answer:aggregation}), which we implemented using the Apache Common Math library. If the error exceeds the error bound target, a feedback mechanism is activated to re-tune the sampling and randomization parameters to provide higher utility in the subsequent epochs.

For the historical analytics, we asynchronously store the (randomized responses) data  in HDFS~\cite{hdfs} 
at the aggregator (as a separate pipeline, which is not shown in Figure~\ref{fig:implementation} for simplicity). To support historical analytics on the stored data at the aggregator,  we also implemented a sampling method {\em sample()} in Flink to support our sampling mechanism ($\S$\ref{subsec:historical}).

\section{Evaluation: Microbenchmarks}
\label{sec:evaluation}

In this section, we evaluate \projecttitle using a series of microbenchmarks. For all microbenchmark measurements, we report the average over $100$ runs. 

\myparagraph{\#I: Effect of sampling and randomization parameters}
\label{subsec:sampling-randomizing}

We first measure the effect of randomization parameters on the utility and the privacy guarantee of the query results.  In particular, the utility is measured by the query results' accuracy loss (Equation~\ref{eq:accuracyloss}), and privacy is measured by the level of achieved zero-knowledge privacy (Equation~\ref{eq:ezk}). For the experiment, we generated $10,000$ original answers randomly, 60\% of which are ``Yes" answers. The sampling parameter $s$ is set to $0.6$.

Table~\ref{tab:utility-privacy} shows that different settings of the two randomization parameters, $p$ and $q$, do affect the utility and the privacy guarantee of the query results.  The higher $p$ means the higher probability that a client responds with its truthful answer.  As expected, this leads to higher utility (i.e., smaller accuracy loss $\eta$) but weaker privacy guarantee (i.e., higher privacy level $\epsilon$).  In addition, Table~\ref{tab:utility-privacy} also shows that the closer we set the probability $q$ to the fraction of truthful ``Yes'' answers (i.e., $60\%$ in this microbenchmark), the higher utility the query result provides. Nevertheless, to meet the utility and privacy requirements in various scenarios, we should carefully choose the appropriate $p$ and $q$. In practice, the selection of the $\epsilon$ value depends on real-world applications~\cite{epsilon-selection}.

We also measured the effect of sampling parameter on the accuracy loss. Figure~\ref{fig:benchmark1} (a) shows that the accuracy loss decreases with the increase of sampling fraction, regardless of the settings of randomization parameters $p$ and $q$.  The benefits reach diminishing returns after the sampling fraction of 80\%. The system operator can set the sampling fraction using resource prediction model~\cite{conductor-ladis-2010, conductor-podc-2010, conductor-nsdi-2012} for any given SLA.

\myparagraph{\#II: Error estimation}
To analyze the accuracy loss, we first measured the accuracy loss caused by sampling and randomized response \emph{separately}. For comparison, we also computed the total accuracy loss after running the two processes in succession as in \privapprox. In this experiment, we set the number of original answers to $10,000$ with $60\%$ of which being ``Yes'' answers. We measure the accuracy loss of the randomized response process by setting the sampling parameter to $100\%$ ($s = 1$) and the randomization parameters $p$ and $q$ to $0.3$ and $0.6$, respectively. Meanwhile, we measure the accuracy loss of the sampling process without the randomized response process by setting $p$ to $1$.

\begin{table}[t!]

\myfontsize
\centering
\caption{Utility and privacy of query results with different randomization parameters $p$ and $q$.  
}

\begin{tabular}{c|c|c|c}
\hline $p$ & $q$ & Accuracy loss ($\eta$) & Privacy Level ($\epsilon$) \\
\hline
\hline \multirow{3}{*}{0.3} & 0.3 & 0.0278 & 1.7047\\
  & 0.6 & 0.0262 & 1.3862\\
  & 0.9 & 0.0268 & 1.2527\\
\hline \multirow{3}{*}{0.6} & 0.3 & 0.0141 & 2.5649\\
  & 0.6 & 0.0128 & 2.0476\\
  & 0.9 & 0.0136 & 1.7917\\
\hline \multirow{3}{*}{0.9} & 0.3 & 0.0098  & 4.1820 \\
  & 0.6 & 0.0079 & 3.5263\\
  & 0.9 & 0.0102 & 3.1570\\
\hline
\end{tabular}
\label{tab:utility-privacy}
\end{table}

Figure~\ref{fig:benchmark1} (b) represents that the accuracy loss during the two experiments is statistically independent to each other. In addition, the accuracy loss of the two processes can effectively be added together to calculate the total accuracy loss.

\begin{figure*}[t]
\centering
\includegraphics [width=1\textwidth]{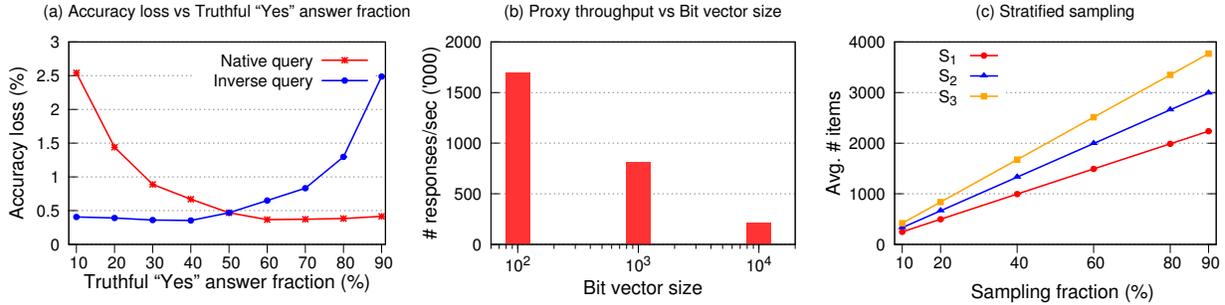}
\caption{{\bf (a)} Accuracy loss for the native and inverse query results with different fractions of truthful ``Yes'' answers. 
{\bf (b)} Throughput of proxies with different bit-vector sizes for the query answer.
{\bf (c)} Average number of sampled data items after stratified sampling with different sampling fractions.}
\vspace{-2mm}
\label{fig:benchmark}
\end{figure*}

%
%
%

\begin{table*}[t!]
\caption{Comparison of crypto overheads (\# operations/sec).  The public-key crypto schemes use a $1024$-bit key.}
\myfontsize
\centering
\begin{tabular}{r|r|r|r|r|r|r|r|r|r|r|r|r}
\hline
                                                                           & \multicolumn{6}{c|}{Encryption}      & \multicolumn{6}{c}{Decryption}                \\ \hline 
                                                                           & \multicolumn{2}{c|}{Phone}    &   \multicolumn{2}{c|}{Laptop}     &  \multicolumn{2}{c|}{Server}         &  \multicolumn{2}{c|}{Phone}          &  \multicolumn{2}{c|}{Laptop}          &  \multicolumn{2}{c}{Server} \\ \hline \hline
RSA~\cite{RSA}           		& 937 & $16\times$    &   	2,770   & $341\times$     	& 	4,909  &      $275\times$   & 	126     &      $25890\times$   & 698    & 	$23666\times$         & 859  	 &      $26401\times$             \\ \hline
Goldwasser~\cite{Chen-2012}   & 2,106  & $7\times$    &  	17,064  & $55\times$   	& 	22,902    &   $59\times$   & 	127        &  $25686\times$   & 6,329   &    $2610\times$      	& 7,068 	 &       $3209\times$          \\ \hline
Paillier~\cite{Paillier}     		& 116  &  $129\times$     	&  	489     &    $1930\times$  	& 	579     &       $2335\times$   & 	72       &      $45308\times$   & 250  &         $66076\times$      	& 309       	 &       $73392\times$       \\ \hline \hline
\projecttitle                     &  \multicolumn{2}{c|}{15,026 } 			&  	 \multicolumn{2}{c|}{943,902 }		 	& 	 \multicolumn{2}{c|}{1,351,937} 		  & 	 \multicolumn{2}{c|}{3,262,186 } 		&  \multicolumn{2}{c|}{16,519,076}  			&  \multicolumn{2}{c}{22,678,285}  \\ \hline
\end{tabular}
\vspace{-4mm}

\label{tab:overhead}
\end{table*}

\myparagraph{\#III: Effect of the number of clients}
\label{subsec:varying-clients}
We next analyzed how the number of participating clients affects the utility of the results.  In this experiment, we fix the sampling and randomization parameters $s$, $p$ and $q$ to $0.9$, $0.9$ and $0.6$, respectively, and set the fraction of truthful ``Yes'' answers to $60\%$.  

Figure~\ref{fig:benchmark1} (c) shows that the utility of query results improves with the increase of the number of participating clients, and few clients (e.g., $<100$) may lead to low-utility query results.

Note that increasing the number of participating clients leads to higher network overheads. However, we can tune the number of clients using the sampling parameter $s$ and thus decrease the network overhead (see $\S$\ref{subsec:cs-varying-clients}).

\myparagraph{\#IV: Effect of the fraction of truthful answers}
\label{subsec:truthful-answers}
We also measured the utility of both the native and the inverse query results with different fractions of truthful ``Yes" answers.  For the experiment, we still keep the sampling and randomization parameters $s$, $p$ and $q$ to $0.9$, $0.9$ and $0.6$, respectively, and set the total number of answers to $10,000$.  

Figure~\ref{fig:benchmark} (a) shows that \privapprox achieves higher utility as the fraction of truthful ``Yes" answers gets closer to $60$\% (i.e., the $q$ value). In addition, when the fraction of truthful ``Yes'' answers $y$ is too small compared to the $q$ value (e.g., $y = 0.1$), the accuracy loss is quite high at $2.54$\%. However, by using the query inversion mechanism ($\S$\ref{subsec:negation}), we can significantly reduce  the accuracy loss to $0.4$\%.

\myparagraph{\#V: Effect of answer's bit-vector sizes}
\label{subsec:num-bitvector}
We measured the throughput at proxies with various bit-vector sizes of client answers (i.e., $A[n]$ in~$\S$\ref{subsec:submitting-queries}). We conducted this experiment with a $3$-node cluster (see $\S$\ref{subsec:experimental-setup} for the experimental setup). 
Figure~\ref{fig:benchmark} (b) shows that the throughput, as expected, is inversely proportional to the answer's bit-vector sizes.

\myparagraph{\#VI: Effect of stratified sampling} To illustrate the use of stratified sampling, we generated a synthetic data stream with three different stream sources $S_1$, $S_2$, $S_3$. Each stream source is created with an independent Poisson distribution. In addition, the three stream sources have an arrival rate of $3:4:5$ data items per time unit, respectively. The computation window size is fixed to $10,000$ data items. 

Figure~\ref{fig:benchmark} (c) shows the average number of selected items of each stream source with varying sample fractions using the stratified sampling mechanism.  As expected, the average number of sampled data items from each stream source is proportional to its arrival rate and the sample fractions.

\myparagraph{\#VII: Computational overhead of crypto operations}
\label{subsec:encryption-overheads}

We compared the computational overhead of crypto operations used in \projecttitle and prior systems.  In particular, these crypto operations are XOR in \privapprox, RSA in~\cite{RSA}, Goldwasser-Micali in~\cite{Chen-2012}, and Paillier in~\cite{Paillier}. For the experiment, we measured the number of crypto operations that can be executed on: {\em (i)} Android Galaxy mini III smartphone running Android 4.1.2 with a 1.5 GHz CPU; {\em (ii)} MacBook Air laptop with a 2.2 GHz Intel Core i7 CPU running OS X Yosemite 10.10.2; and {\em (iii)} Linux server running Linux 3.15.0 equipped with a 2.2 GHz CPU with 32 cores. 

Table~\ref{tab:overhead} shows that the XOR operation is extremely efficient compared with the other crypto mechanisms. This highlights the importance of XOR encryption in our design.

\begin{table}[]
\centering

\caption{Throughput (\# operations/sec) at clients}
\label{client-overhead}
\myfontsize
\begin{tabular}{@{}l|l|l|l@{}}
\hline
              No. of operations/sec                         & Phone      & Laptop    & Server                         \\  \hline
\hline
SQLite read                        & 1,162      & 19,646   & 23,418                        \\  \hline
Randomized response		    & 168,938  & 418,668 & 1,809,662                   \\  \hline
XOR encryption                   & 15,026    & 943,902 & 1,351,937                   \\  \hline
{\bf Total  }                                & 1,116      & 17,236   &  22,026                       \\  \hline
\end{tabular}
\end{table}
\myparagraph{\#VIII: Throughput at clients} We measured the throughput at clients.  In particular, we measured the number of operations per second that can be executed at clients for the query answering process. In this experiment, we used the same set of devices as in the previous experiment. 

Table~\ref{client-overhead} presents the throughput at clients. To further investigate the overheads, we measured the individual throughput of three sub-processes in the query answering process: {\em (i)} database read, {\em (ii)} randomized response, and {\em (iii)} XOR encryption. The result indicates that the performance bottleneck in the answering process is actually the database read operation.

\myparagraph{\#IX: Comparison with related work}
\label{subsec:splitx-comparison}
First, we compared \privapprox with SplitX~\cite{Chen-2013}, a high-performance privacy-preserving analytics system. We compare the latency incurred at proxies in both \projecttitle and SplitX. SplitX is geared towards batch analytics, but can be adapted to enable privacy-preserving data analytics over data streams.  Since \projecttitle and SplitX share the same architecture, we compare the latency incurred at proxies in both systems.  

Figure~\ref{fig:splitx-comparison} shows that, with different numbers of clients, the latency incurred at proxies in \projecttitle is always nearly one order of magnitude lower than that in SplitX.  The reason is simple: unlike \projecttitle, SplitX requires synchronization among its proxies to process query answers in a privacy-preserving fashion.  This synchronization creates a significant delay in processing query answers, making SplitX unsuitable for dealing with large-scale stream analytics. More specifically, in SplitX, the processing at proxies consists of a few sub-processes including adding noise to answers, answer transmission, answer intersection, and answer shuffling; whereas, in \privapprox, the processing at proxies contains only the answer transmission.  Figure~\ref{fig:splitx-comparison} also shows that with $10^6$ clients, the latency at SplitX is $40.27$ sec, whereas \projecttitle achieves a latency of just $6.21$ sec, resulting in a $6.48\times$ speedup compared with SplitX.

\begin{figure}[t]
\centering
\includegraphics [width=0.38\textwidth]{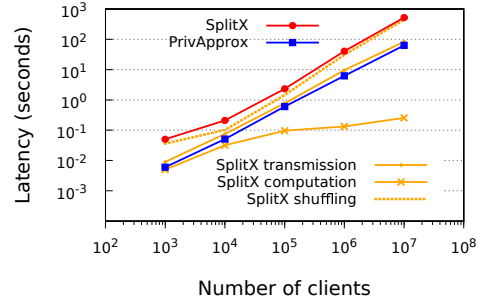}
\caption{Latency comparison b/w SplitX and \privapprox.}
\label{fig:splitx-comparison}
\end{figure}

Next, we compared \privapprox with a recent privacy-preserving analytics system called RAPPOR~\cite{ErlingssonPK14}. Similar to \projecttitle, RAPPOR applies a randomized response mechanism to achieve differential privacy.  However, RAPPOR is not designed for stream analytics, and therefore, we compared \projecttitle with RAPPOR for privacy only. 
To make an ``apple-to-apple" comparison between \projecttitle and RAPPOR in terms of privacy, we make a mapping between the system parameters of the two systems. We set the sampling parameter $s = 1$, and the randomized parameters $p = 1- f$, $q = 0.5$ in \projecttitle, where $f$ is the parameter used in the randomized response process of  RAPPOR~\cite{ErlingssonPK14}. In addition, we set the number of hash functions used in RAPPOR to $1$ ($h = 1$) for a fair comparison. In doing so, the two systems have the same randomized response process.  However, since \projecttitle makes use of the sampling mechanism before performing the randomized response process, \projecttitle achieves stronger privacy. Figure~\ref{fig:rappor-comparison} shows the differential privacy level of RAPPOR and \projecttitle with different sampling fractions $s$.

It is worth mentioning that, by applying the sampling mechanism, \projecttitle achieves stronger privacy (i.e., zero\-/knowledge privacy) for clients. The comparison between differential privacy and zero\-/knowledge privacy is presented in the Appendix~\ref{sec:privacy-evaluation}.

\begin{figure}[t]
\centering
\includegraphics [width=0.38\textwidth]{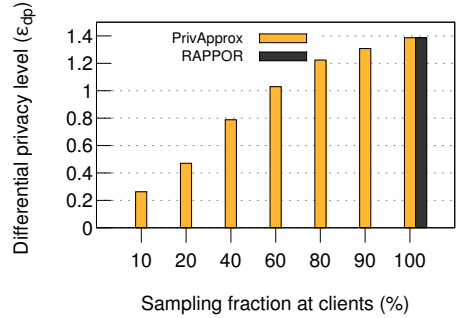}
\caption{Differential privacy level comparison b/w RAPPOR and \privapprox.}
\label{fig:rappor-comparison}
\end{figure}

Recently, several privacy-preserving stream analytics systems have been proposed~\cite{rescuedp,bes-stream,lucas-stream}. These systems make use of the Laplace mechanism~\cite{dwork-tcc,differential-privacy} to achieve differential privacy. In particular, they add Laplace noise to the truthful answers {\em at the aggregator} to protect  the users' privacy. However, their approach relies on strong trust assumptions of the aggregator as well as the connection between clients and the aggregator. On the contrary, \projecttitle applies randomized response mechanism to process users' private data locally at clients under the control of users. Combined with the sampling mechanism, \projecttitle achieves  stronger privacy guarantees (with a tighter bound for \(\epsilon_{dp}\)\-/differential privacy and  \(\epsilon_{zk}\)\-/zero\-/knowledge privacy). 

\section{Evaluation: Case-studies}
\label{sec:case-studies}

We next present our experience using \projecttitle in the following two case studies: {\em (i)} New York City (NYC) taxi ride, and {\em (ii)} household electricity consumption.

\subsection{Experimental Setup}
\label{subsec:experimental-setup}
\myparagraph{Cluster setup} We used a cluster of $44$ nodes connected via a Gigabit Ethernet. Each node contains 2 Intel Xeon quad-core CPUs and 8 GB of RAM running Debian 5.0 with Linux kernel 2.6.26.  We deployed two proxies with Apache Kafka, each of which consists of $4$ Kafka broker nodes and $3$ Zookeeper nodes. We used  $20$ nodes to deploy Apache Flink as the aggregator. In addition, we employed the remaining $10$ nodes to replay the datasets to generate data streams for evaluating our \privapprox system.

\myparagraph{Datasets} For the first case study, we used the {\em NYC Taxi Ride} dataset from the DEBS 2015 Grand Challenge~\cite{nyc-taxi-dataset}. The dataset consists of the itinerary information of all rides across $10,000$ taxies in New York City in 2013.
For the second case study, we used the {\em Household Electricity Consumption} dataset~\cite{electricity-dataset}.  This dataset contains electricity usage (kWh) measured every 30 minutes for one year by smart meters.

\myparagraph{Queries} For the NYC taxi ride case-study, we created a query: ``{\em What is the distance distribution of taxi trips in New York?}". We defined the query answer with $11$ buckets as follows:  [0, 1) mile, [1, 2) miles, [2, 3) miles, [3, 4) miles, [4, 5) miles, [5, 6) miles, [6, 7) miles, [7, 8) miles, [8, 9) miles, [9, 10) miles, and [10, $+\infty$) miles.

For the second case-study, we defined a query to analyze the electricity usage distribution of households over the past 30 minutes. The query answer format is as follows: [0, 0.5] kWh, (0.5, 1] kWh, (1, 1.5] kWh, (1.5, 2] kWh, (2, 2.5] kWh, and (2.5, 3] kWh.

\myparagraph{Evaluation metrics} We evaluated our system using four key metrics: throughput, latency, utility, and privacy level. {\em Throughput} is defined as the number of data items processed  per second, and {\em latency} is defined as the total amount of time required to process a certain dataset. {\em Utility}  is the accuracy loss defined as $| \frac{estimate - exact}{exact} |$, where $estimate$ and $exact$ are the query results produced by applying \privapprox and the native computation, respectively. Finally, {\em privacy level} (\(\epsilon_{zk}\)) is calculated  using equation~\ref{eq:ezk}.  For all measurements, we report the average over $10$ runs.

\begin{figure}[t]
\centering
\includegraphics [width=0.49\textwidth]{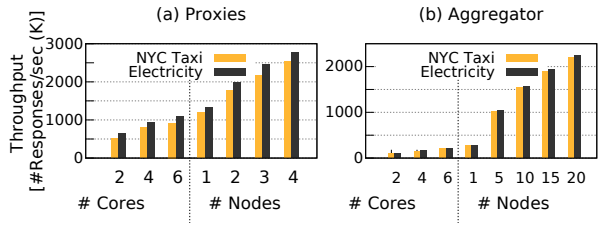}
\caption{Throughput at proxies and the aggregator with different numbers of CPU cores and nodes.}
\label{fig:throughput-proxies-aggregator}
\end{figure}

\subsection{Results from Case-studies}

\begin{figure*}[t]
\centering
\includegraphics [width=1\textwidth]{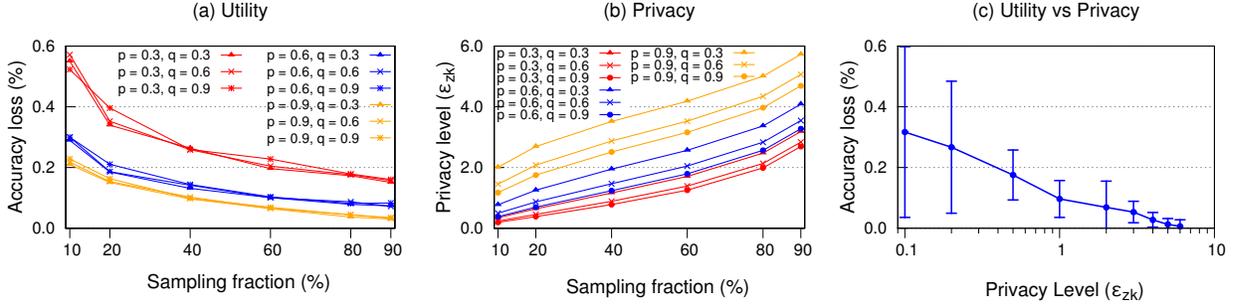}
\caption{Results from the NYC taxi case-study with varying sampling and randomization parameters: {\bf (a)}  Utility, {\bf (b)} Privacy level, {\bf (c)} Comparison between utility and privacy.} 
\vspace{-5mm}
\label{fig:utility-privacy-taxi}
\end{figure*}

\subsubsection{Scalability}
\label{subsec:cs-scalability}
We measured the scalability of the two main system components: proxies and the aggregator. We first measured the throughput of proxies with various numbers of CPU cores (scale-up) and different numbers of nodes (scale-out). This experiment was conducted on a cluster of $4$ nodes. Figure~\ref{fig:throughput-proxies-aggregator} (a) shows that, as expected, the throughput at proxies  scales quite well with the number of CPU cores and
nodes. In the NYC Taxi case-study, with $2$ cores, the throughput of each proxy is $512,348$ answers/sec, and with $8$ cores (1 node) the throughput is $1,192,903$ answers/sec; whereas, with a cluster of $4$ nodes each with $8$ cores, the throughput of each proxy reaches $2,539,715$ answers/sec. In  the household electricity case-study, the proxies achieve relatively higher throughput because the message size is smaller than in the NYC Taxi case-study.

We next measured the throughput at the aggregator. 
Figure~\ref{fig:throughput-proxies-aggregator} (b) depicts that the aggregator also scales quite well when the number of nodes for aggregator increases. The throughput of the aggregator, however, is much lower than the throughput of  proxies due to the relatively expensive {\tt join} operation and the analytical computation at the aggregator.  We notice that the throughput of the aggregator in the household electricity case study does not significantly improve in comparison to the first case study.  This is because the difference in the size of messages between the two case studies does not affect much the performance of the {\tt join} operation and the analytical computation.

\subsubsection{Network Bandwidth and Latency}
\label{subsec:cs-varying-clients}

Next, we conducted the experiment to measure the network bandwidth usage. By leveraging the sampling mechanism at clients, our system reduces network traffic  significantly. Figure~\ref{fig:bandwidth-latency} (a) shows the total network traffic transferred from clients to proxies with different sampling fractions. In the first case study, with the sampling fraction of $60$\%, \projecttitle can reduce the network traffic by $1.62\times$; whereas in the second case study, the reduction is $1.58\times$.

Besides the benefit of saving network bandwidth, \projecttitle achieves also lower latency in processing query answers by leveraging approximate computation. To evaluate this advantage, we measured the  effect of sampling fractions on the latency of processing query answers. Figure~\ref{fig:bandwidth-latency} (b) depicts the latency with different sampling fractions at clients. For the first case-study, with the sampling fraction of $60$\%, the latency is $1.68\times$
lower than the execution without sampling; whereas, in the second case-study this value is $1.66 \times$ lower  than the execution without sampling.

\begin{figure}[t]
\centering
\includegraphics [width=0.49\textwidth]{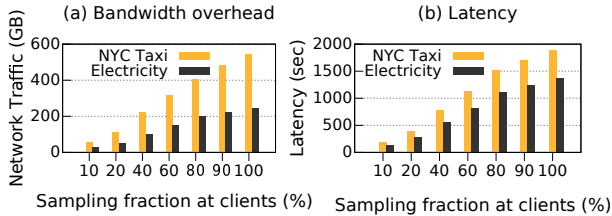}
\caption{Total network traffic and latency at proxies with different sampling fractions at clients.}
\label{fig:bandwidth-latency}
\end{figure}

\subsubsection{Utility and Privacy}
\label{subsec:cs-sampling-randomizing}

Figure~\ref{fig:utility-privacy-taxi} (a)(b)(c) show the utility, the privacy level, and the trade-off between them, respectively, with different sampling and randomization parameters.  The randomization parameters $p$ and $q$ vary in the range of (0, 1), and the sampling parameter $s$ is calculated using  Equation~\ref{eq:ezk}. Here, we show results only for NYC Taxi dataset. As the sampling parameter $s$ and the first randomization parameter $p$ increase, the utility of query results improves (i.e., accuracy loss gets smaller) whereas the privacy guarantee gets weaker (i.e., privacy level gets higher). Since the New York taxi dataset is diverse, the accuracy loss and the privacy level change in a non-linear fashion with different sampling fractions and randomization parameters.
Interestingly, the accuracy loss does not always decrease as the second randomization parameter $q$ increases. The accuracy loss gets smaller when $q = 0.3$.  This is due to the fact that the fraction of truthful ``Yes" answers in the dataset is $33.57$\% (close to $q=0.3$).

\subsubsection{Historical Analytics}
\label{subsec:cs-historical-analytics}
To analyze the performance of \projecttitle for historical analytics, we executed the queries on the datasets stored at the aggregator. Figure~\ref{fig:historical-analysis} (a) (b) present the latency and throughput, respectively, of processing historical datasets with different sampling fractions. We can achieve a speedup of $1.86 \times$ over native execution in historical analytics by setting the sampling fraction to $60$\%. 

We also measured the accuracy loss when the approximate computation was applied (for the NYC Taxi case-study).  Figure~\ref{fig:historical-analysis} (c) shows the accuracy loss in processing historical data with different sampling fractions. With the sampling fraction of $60$\%, the accuracy loss is only less than $1$\%.

\begin{figure*}[t]
\centering
\includegraphics [width=1\textwidth]{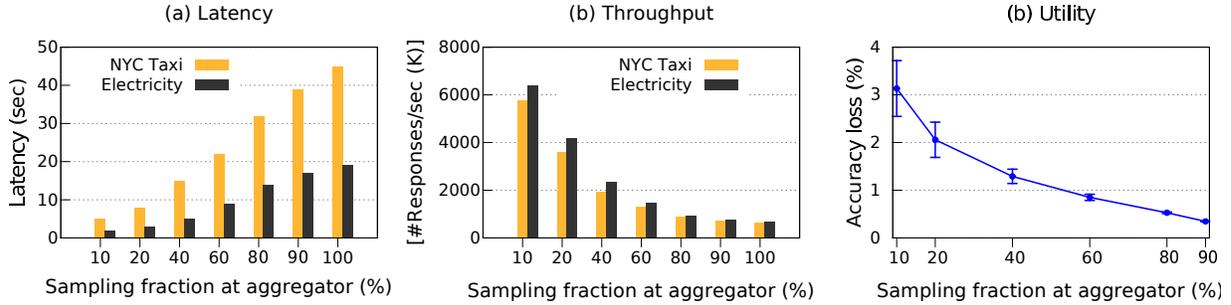}
\caption{Historical analytics results with varying sampling fractions: {\bf (a)} Latency, {\bf (b)} Throughput, and {\bf (c)} Utility.}
\vspace{-5mm}
\label{fig:historical-analysis}
\end{figure*}

\section{Related Work}
\label{sec:related}

\myparagraph{Privacy-preserving analytics} Since the notion of differential privacy~\cite{differential-privacy, dwork-tcc}, a plethora of systems have been proposed to provide differential privacy with  centralized trusted databases supporting linear queries~\cite{linear-queries}, graph queries~\cite{graph-queries}, histogram queries~\cite{histogram-queries}, Airavat-MapReduce~\cite{airavat}, SQL-type PINQ queries~\cite{pinq, frank-ratul-network-privacy, goldberg-wPINQ} and even general programs, such as GUPT~\cite{gupt} and Fuzz~\cite{fuzz}. In practice, however, such central trust can be abused, leaked, or subpoenaed~\cite{violation1,violation2,violation3,violation4}.

To overcome the limitations of the centralized database schemes, recently a flurry of systems have been proposed with a 
focus on achieving users' privacy (mostly, differential privacy) in a distributed setting where the private data is kept locally.  Examples include Privad~\cite{privad}, PDDP~\cite{Chen-2012}, DJoin~\cite{djoin}, SplitX~\cite{Chen-2013}, $\pi$Box~\cite{pi-box},  KISS~\cite{KISS}, Koi~\cite{koi}, xBook~\cite{xBook}, Popcorn~\cite{popcorn}, and many other systems~\cite{distributed-noise,mobile-ad, RSA}.
However, these  systems are designed to deal with the ``one-shot" batch queries only, whereby the data is assumed to be static during the query execution. 

To overcome the limitations of the aforementioned systems, several differentially private stream analytics systems have been proposed recently~\cite{differential-privacy-obs, private-statistic, privacy-monitor, privacy-time-series, Paillier, stream-monitor, stream-privacy-fault}.  
Unfortunately, these systems still contain several technical shortcomings that limit their practicality.  One of the first systems~\cite{differential-privacy-obs} updates the query result only if the user's private data changes significantly, and does not support stream analytics over an unlimited time period.  Subsequent systems~\cite{private-statistic, stream-privacy-fault} remove the limit on the time period, but introduce extra system overheads.  Some
systems~\cite{privacy-time-series, Paillier} leverage expensive secret sharing cryptographic operations to produce noisy aggregate query results.  These protocols, however, cannot work at large scale under churn; moreover, in these systems, even a single malicious user can substantially distort the aggregate results without detection.  Recently, some other privacy-preserving distributed stream monitoring systems have been proposed~\cite{stream-monitor, privacy-monitor}.
However, they all require some form of synchronization, and are tailored for heavy-hitter monitoring only. Streaming data publishing systems like~\cite{rescuedp} use a stream-privacy metric at the cost of relying on a trusted party to add noise.
In contrast,  \privapprox does not require a trusted proxy or aggregator to add noise. Furthermore, \privapprox  provides stronger privacy properties (zero-knowledge privacy).

\myparagraph{Sampling and randomized response}  Sampling and randomized response, also known as {\em input perturbation} techniques, are being studied in the context of privacy-preserving analytics, albeit they are explored separately. For instance,  the relationship between {\em sampling} and privacy is being investigated to provide k-anonymity~\cite{sampling-privacy1}, differential privacy~\cite{gupt}, and crowd-blending privacy~\cite{crowd-blending}.
In contrast, we show that sampling combined with randomized response achieves the zero-knowledge privacy, a privacy bound strictly stronger that the state-of-the-art differential privacy. Furthermore, \privapprox achieves these guarantees for stream processing with a distributed private dataset.

{\em Randomized response}~\cite{fox1986randomized, randomized-response} is a surveying technique in statistics, since 1960s, for collecting sensitive information via input perturbation. Recently, Google, in a system called RAPPOR~\cite{ErlingssonPK14}, made use of randomized response for privacy-preserving analytics for the Chrome browser. RAPPOR provides differential privacy  (\(\epsilon_{dp}\)) for clients while enabling analysts to collect various types of statistics. Like RAPPOR,  \privapprox utilizes randomized response. However, RAPPOR is  designed for heavy-hitter collection, and does not deal with the situation where clients' answers to the same query are changing over time. Therefore, RAPPOR does not fit well with the stream analytics. Furthermore, since we combine randomized response with sampling, \privapprox (\(\epsilon_{zk}\)) provides a privacy bound tighter than RAPPOR (\(\epsilon_{dp}\)). 

\myparagraph{Secure multi-party computation} In theory, secure multi-party
computation (SMC)~\cite{GoldreichMW87,Yao82b} could be used  for privacy-preserving analytics.
It is, however, expensive for real-world deployment, especially for stream analytics,
even though there have been several proposals reducing
SMC's computational overhead~\cite{GordonMRW13,JareckiS07,LindellP07,LindellP15,PinkasSSW09,Woodruff07}.
Furthermore, SMC guarantees input-privacy during computation, but is orthogonal to output-privacy as provided by differential privacy.

\myparagraph{Approximate computing} Approximation techniques such as sampling~\cite{stratified-sampling, sampling-2, Chaudhuri-sampling}, sketches~\cite{sketching}, and online aggregation~\cite{online-aggregation} have been well-studied over the decades in the databases community. Recently, sampling-based systems (such as ApproxHadoop~\cite{approxhadoop}, BlinkDB~\cite{BlinkDB, BlinkDB-2}, IncApprox~\cite{incapprox-www-2016}, Quickr~\cite{quickr-sigmod}, StreamApprox~\cite{streamapprox})  and online aggregation-based systems (such as MapReduce Online~\cite{mapreduce-online, OAMR}, G-OLA~\cite{gola}) have also been shown effective for ``Big Data" analytics.

We  build on the advancements of sampling-based techniques.  However, we differ in two crucial aspects. First, we perform sampling in a distributed way as opposed to sampling in a centralized dataset.  Second, we extend sampling with randomized response for privacy-preserving analytics.

\section{Conclusion}
\label{sec:conclusion}

In this technical report, we presented \privapprox, a privacy-preserving stream analytics system. Our approach builds on the observation that both computing paradigms --- privacy-preserving data analytics and approximate computation --- strive for approximation, and can be combined together to leverage the benefits of both. Our evaluation shows that \projecttitle not only improves the performance to support real-time stream analytics, but also achieves provably stronger privacy guarantees than the state-of-the-art differential privacy.  This technical report is the complete version of our conference publication~\cite{privapprox-atc17}. \privapprox's source code is publicly available: \href{https://privapprox.github.io}{https://PrivApprox.github.io}.

\section*{Appendices}
\appendix
\addcontentsline{toc}{section}{APPENDICES}
\section{Algorithms}
\label{sec:algorithms}

\begin{algorithm}[t]

\myfontsize
\SetLine

\textbf{Input}:  Query and query budget\\  

%

$s, p, q \leftarrow $  {\tt costFunction}({\em budget});  {\commentfontsize  // {\em $s$ is the sampling parameter}}\\
{\commentfontsize // {\em $p$ and $q$ are the randomizing parameters}}\\

%
%
%

$A[n] \leftarrow Query.A[n]$;  {\commentfontsize  // {\em Answer bit-vector }}\\

{\bf \underline {execute---At---Client}}() {\commentfontsize  //Execute the method every $f$ seconds}\\
\Begin{
     $flipResult0 \leftarrow$ {\tt coinFlip}($s$);{\commentfontsize  // {\em Flip the sampling coin}}\
     
     \If{$flipResult0 == $``{\tt Heads}"} {
         $client.participate \leftarrow$ ``{\tt True}" \;
         $truthfulAnswer$  $\leftarrow$ {\tt localDataProcess}($Query.SQL$)\; 
         $flipResult1 \leftarrow$ {\tt coinFlip}($p$);  {\commentfontsize  // {\em First randomizing coin}}\\
         
         \If{$flipResult1 ==$``{\tt Heads}"} {
             $A[n]  \leftarrow$  $truthfulAnswer$;  {\commentfontsize  // {\em Process the local data}}\\
         
         }
         \Else {
             $flipResult2 \leftarrow$ {\tt coinFlip}($q$);   {\commentfontsize  // {\em Second coin}}\\
             
             \If{$flipResult2 ==$``{\tt Heads}"} {
                  {\commentfontsize  // for all {\bf ``Yes"} in the bit-vector}\\ 
                  $\forall i\in \{1, ...,n\}$:   {\bf if}$(A[i] == 1)$ $A[i]  \leftarrow 1$; 
             }
             \Else {
                 	{\commentfontsize  // for all {\bf ``No"} in the bit-vector}\\ 
                                   $\forall i\in \{1, ...,n\}$:  {\bf if}$(A[i] == 1)$ $A[i]  \leftarrow 0$; 
		                  
             }
         
         }
         
         {\tt sendAnswer}($A[n]$); {\commentfontsize // {\em Send the answer to the aggregator}}\\
     }
}

\caption{\bf Answering a query at clients}
\label{alg:client-algo}
\end{algorithm}
 
In this section, we describe the algorithmic details of \privapprox's system protocol. We present two algorithms: {\em (i)} the workflow at a client carrying out sampling and randomization; and {\em (ii)} the workflow at the aggregator.

\myparagraph{\#I: Workflow at a client} Algorithm~\ref{alg:client-algo} summarizes how a client processes a query.  Each client maintains its personal data in a local database. Upon receiving a query, the client first flips a sampling coin to decide whether to answer the query or not. If the coin comes up heads, then the client executes the query on its local database to create a truthful answer to the query. The truthful answer is in the form of bit buckets with a ``1'' or ``0'' per bucket, depending on whether or not the ``Yes'' answer falls within that bucket. The answer may have more than one bucket containing a ``1'' depending on the query.  Next, the client randomizes the answer using the randomized response mechanism. In particular, the client flips the first randomization coin, if it comes up heads, the client responds its truthful answer. If it comes up tails, then the client flips the second randomization coin and reports the result of this coin flipping. The randomized answer is still in the binary string format after the randomization process.

\myparagraph{\#II: Workflow at the aggregator}
The aggregator receives clients' data streams from the proxies, and joins them to obtain a combined data stream. Thereafter, the aggregator processes the joined stream to produce the output for the analyst. Algorithm~\ref{alg:aggregator-algo2} describes the overall process at the aggregator. The algorithm computes the query results as a sliding window computation over the incoming answer stream. For each window, the aggregator first adapts the computation window to the current start time $t$ by removing all old data items, i.e., with $timestamp < t$, from the computation window. Next, the aggregator adds the new incoming data items in the window and decrypts the answers in the data stream. Thereafter, the input data items for a window are aggregated to produce the query output for the analyst. We also estimate the error in the output due to approximation and randomization. The aggregator estimates this error bound and defines a confidence interval for the result as: $queryResult \pm errorbound$. The entire process is repeated for the next window, with the updated windowing parameters and query budget (for the adaptive execution).

 \begin{algorithm}[t]

\myfontsize
\SetLine

\textbf{Input}:  
$w \leftarrow query.time\_window$; 
$\delta \leftarrow query.slide\_interval$; \\
$t$ $\leftarrow$ start time of window;\\ 
{\bf \underline {execute---At---Aggregator}}() {\commentfontsize  //Execute the method every $\delta$ seconds}\\
\Begin{
$window$ $\leftarrow$  $\emptyset$;   {\commentfontsize  // {\em List of items in the window}}\\
    
	\ForEach{(window $w$ in the incoming stream )}{
		
		\ForAll{$elements$ in $w$}{
			\If{$element$.timestamp $<$ $t$}{
				$w$.{\tt remove}($element$); {\commentfontsize // {\em Remove all old items}}\\
			}
		}		
		$w$ $\leftarrow$ $w$.{\tt insert}({\em new items}); {\commentfontsize // {\em Add new items}}\\
		
		
		$queryResult$ $\leftarrow$  $\emptyset$;   {\commentfontsize  // {\em query result}}\\
		\ForAll{$answer$ in the $sample$}{		
		
              $queryAnswer$ $\leftarrow$ {\tt decryptAnswer($answer$)}\;      
              
              {\commentfontsize // {\em Get query results associate with analyst IDs}}\\
              $queryResult$ $\leftarrow$ {\tt aggregateAnswer($queryAnswer$)}\;     
		 }	 
		 
		$queryResult \pm error$ $\leftarrow$ {\tt estimateError}($queryResult$)\; 
		 
		$t$ $\leftarrow$ $t$ + $\delta$; {\commentfontsize // {\em Update the start time for the next window}}\\
   }
}

\caption{\bf Generating query result at the aggregator}
\label{alg:aggregator-algo2}
\end{algorithm}
\section{Discussion}
\label{sec:discussion}
In this section, we discuss some approaches that could be used to meet our assumptions listed in $\S$\ref{sec:design-assumptions}.

\myparagraph{Stratified sampling} In our design in $\S$\ref{sec:design}, we currently assume that the input stream is already stratified based on the source of event, i.e., the data items within each stratum follow the same distribution. However, it may not be the case.  We next discuss two proposals for the stratification of evolving data streams, namely bootstrap~\cite{bootstrap-Dziuda,bootstrap-efron,bootstrap-Odile} and semi-supervised learning~\cite{semi-supervised-algorithm}. 

Bootstrap~\cite{bootstrap-Dziuda,bootstrap-efron,bootstrap-Odile} is a well-studied non-parametric sampling technique in statistics for the estimation of distribution for a given population. In particular, the bootstrap method randomly selects ``bootstrap samples" with replacement to estimate the unknown parameters of a population; for instance, by averaging the bootstrap samples. We can employ a bootstrap-based estimator for the stratification of incoming sub-streams. Alternatively, we could also make use of a semi-supervised algorithm~\cite{semi-supervised-algorithm} to stratify a data stream. The advantage of the algorithm is that it can work with both the labeled and unlabeled data stream to train a classification model.

\myparagraph{Virtual cost function} Currently, in our implementation described in $\S$\ref{sec:implementation}, for a given user-specified query budget about privacy $\epsilon_{zk}$, the sampling and randomizing parameters can be computed using the reversed function of equation~\ref{eq:ezk}. However, for the query budget involving available computing resources or latency requirements (SLAs)---we currently assume that there exists a virtual function that determines the sampling parameter based on the query budget. We recommend two existing approaches---Pulsar~\cite{pulsar}  and resource prediction model~\cite{resource-prediction-SML,resource-prediction-models}---to design and implement such a virtual function for the given computing resources and latency requirements, respectively.

Pulsar \cite{pulsar} is a ``virtual datacenter (VDC)"  system that allows users to allocate resources based on tenants' demand.  The system proposes a multi-resource token bucket algorithm that uses a pre-advertised cost model for supporting workload independent guarantees.  We could apply a similar cost model based on Pulsar as follows: A data item to be processed could be considered as a request, and ``amount of resources" needed to process these items could be the cost in tokens. Since the resource budget gives total resources (\textit{here tokens}) to be used, we could find the number of clients, i.e., the sampling fraction at clients, that can be processed using these resources. 

To find the sampling parameter for a given latency budget, we could use a resource prediction model~\cite{conductor-nsdi-2012, conductor-podc-2010, conductor-ladis-2010}. The resource prediction model could build by analyzing the diurnal patterns of resource usage~\cite{googlecluster}, and predicts the resource requirement to meet SLAs leveraging statistical machine learning~\cite{resource-prediction-SML,resource-prediction-models}. Once we have the resource requirement  in place to meet a given SLA---we can find the appropriate sampling parameter by using the above suggested method similar to Pulsar.

\section{Privacy Analysis and Proofs}
\label{sec:privacy-evaluation}

\projecttitle achieves three privacy properties {\em (i)}~zero-knowledge privacy, {\em (ii)}~anonymity, and {\em (iii)}~unlinkability as introduced in $\S$\ref{subsec:privacy-properties}. 



\myparagraph{{Property \# I: }{Zero\-/knowledge privacy}}  We show that the system designed in Section~\ref{sec:design} achieves \(\epsilon_{zk}\)\-/zero\-/knowledge privacy and prove a tighter bound for \(\epsilon_{dp}\)\-/differential privacy, than what generally follows from zero\-/knowledge privacy~\cite{zero-knowledge-privacy}. The basic idea is that all data from the clients is already differentially private due to the use of randomized response. Furthermore, the combination with pre-sampling at the clients makes it zero\-/knowledge private as well. Following the privacy definitions, any computation upon the results of differentially, as well as, zero\-/knowledge private algorithms is guaranteed to be private again.

In the following paragraphs we show that:
\begin{itemize}
  \item Independent and identically distributed (IID) sampling decomposes easily and is self\-/commutative. See Lemma~\ref{lem:sampling-commutes}.
  \item Sampling and randomized response mechanisms commute. See Lemma~\ref{lem:sampling-randomizing-commutes}.
  \item Pre\-/sampling and post\-/sampling can be traded arbitrarily arou\-nd a randomized response mechanism. See Corollary~\ref{cor:arbitrary-sampling}.
  \item A \(\epsilon_{zk}\)\-/zero\-/knowledge privacy bound for our system. See Theorem~\ref{thrm:zkp}
  \item A \(\epsilon_{dp}\)\-/differential privacy bound for our system. See Theorem~\ref{thrm:dp}
  \item Our differential privacy bound is tighter than the general differential privacy bound derived from a zero\-/knowledge private algorithm. See Proposition~\ref{prop:tighter-dp}.

%
%
\end{itemize}


Intuitively, differential privacy limits the information that can be learned about any individual \(i\) by the difference occurring from either including \(i\)'s sensitive data in a differentially private computation or not. Zero\-/knowledge privacy on the other hand also gives the adversary access to aggregate information about the remaining individuals denoted as \(D_{-i}\). Essentially everything that can be learned about individual \(i\) can also be learned by having access to some aggregate information upon \(D_{-i}\). 




Let \(San()\) be a sanitizing algorithm, which takes a database \(D\) of sensitive attributes \(a_i\) of individuals \(i \in \mathrm{P}\) from a population \(\mathrm{P}\) as input and outputs a differentially private or zero\-/knowledge private result \(San(D)\). For brevity, we write \(San_A(D)\) for the output of the adversary \(A\) with arbitrary external input \(z\) and access to \(San(D)\). Similarly, we omit the explicit usage of the external information \(z\) as input to the simulator \(S\), as well as the total size of the database. See~\cite{crowd-blending} Definition 1 and 2 for the extended notation. Let \(O \subseteq Range(San_A)\) be any set of possible outputs.
\(\epsilon_{dp}\)-differential privacy can be defined as
\begin{equation}\label{eq:differential-privacy}
  \Pr[San_A(D) \in O] \leq e^{\epsilon_{dp}} \cdot \Pr[San_A(D_{-i}) \in O]
\end{equation}
while \(\epsilon_{zk}\)-zero\-/knowledge privacy is defined as
\begin{equation}\label{eq:zero-knowledge-privacy}
  \Pr[San_A(D) \in O] \leq e^{\epsilon_{zk}} \cdot \Pr[S(T(D_{-i}),|D|) \in O].
\end{equation}


Before proving the desired properties, we need to introduce some notation. Let \(D = \{a_i\}\) be a database of sensitive attributes of individuals \(i \in \mathrm{P}\). For ease of presentation and without loss of generality we restrict the individual's sensitive attribute to a boolean value \(a_i \in \{0,1\}\) and \(\mathrm{D}={a_{i^{\prime}}}\) for all \(i^{\prime} \in \mathrm{P}\). Furthermore, let \(\mathsf{D}(\mathrm{D})=\{U:U\subseteq \mathrm{D}\}\) be the super-set of all possible databases and \(Sam(D,u)\from \- \mathsf{D}(\mathrm{D}) \times (0,1) \to \mathsf{D}(\mathrm{D}) \) be a randomized algorithm that i.i.d. samples rows or individuals with their sensitive attributes from database \(D\) with probability \(s\) without replacement. Let \(San(D,p,q) \from \mathsf{D}(\mathrm{D}) \times (0,1) \times (0,1) \to \mathsf{D}(\mathrm{D}) \) be a two\-/coin randomized response algorithm that decides for any individual \(i^{\prime}\) in database \(D\) with probability \(p\) if it should be part of the output. If it is not included in the output, the result of tossing a biased coin (coming up heads with probability \(q\)) is added to the output. 

\begin{lemma}
  \emph{(Decompose and commute sampling)}
  \label{lem:sampling-commutes}
  Let \(s = uv\) with \(s,u,v \in (0,1)\) being sampling probabilities for a sampling function \(Sam()\). It follows that \(Sam()\) can be composed and decomposed easily and is self\-/commutative.
  \begin{align*}
    Sam(D,s)&\approx Sam(Sam(D,u),v) \\
    &\approx Sam(Sam(D,v),u).
  \end{align*}
\end{lemma}

\begin{proof}
  Let \(Sam_u,Sam_v\) be sampling algorithms that sample rows i.i.d. from a database with probability \(u\) and \(v\) respectively. By applying \(Sam_u(D)\), any row in \(D\) has probability \(u\) of being sampled. The probability for any row in \(D\) being sampled by \(Sam_v\) is equivalently \(v\). Using function composition the probability for any row in \(D\) being sampled by \(Sam_s = (Sam_u \circ Sam_v)(D)\) is
  \begin{equation}\label{eq:sampling-decomposition}
    s=uv.
  \end{equation}
  From multiplication being commutative (\(u\cdot v=v \cdot u\)) follows that \(Sam_u\) and \(Sam_v\) commute, that is \(Sam_u \circ Sam_v = Sam_v \circ Sam_u\). This is true for deterministic functions and can easily be extended to randomized functions described as random variables, as random variables are commutative under addition and multiplication. For ease of presentation and without loss of generality we keep the notion of functions instead of random variables.
  Let \(Sam_s(D)=Sam(D,s)\) be a sampling function that samples rows i.i.d. from a given database \(D\) with probability \(s\). Decomposing sampling function \(Sam_s()\) with probability \(s\) into two functions with probabilities \(u\) and \(v\) follow from \eqref{sampling-decomposition}. It also follows that two sampling functions with probabilities \(u,v\) can be composed into a single sampling function with sampling probability \(s\).
\end{proof}

\begin{lemma}
  \emph{(Commutativity of sampling and randomized response)}
  \label{lem:sampling-randomizing-commutes}
  Given a sampling algorithm \(Sam()\) and a randomized response algorithm \(San()\),
  the result of the pre-sampling algorithm \(F_{pre}(D,s,p,q)=San(Sam(D,s),p,q)\) is statistically indistinguishable from the result of the post-sampling algorithm \(F_{post}(D,s,p,q)=Sam(San(D,p,q),s)\). It follows that sampling and randomized response commute under function composition: \(Sam \circ San = San \circ Sam\).
\end{lemma}
\begin{proof}
  For any individual \(i\) having \(a_i \in D\) we have to consider eight different possible cases. In case the sampling algorithm \(Sam()\) decides to not sample \(i\), it obviously doesn't matter if it gets removed before the randomized response algorithm is run of afterwards. We thus condition on \(Sam()\) to include \(i\) in the output.
  \begin{enumerate}
    \item Let us first consider the case that \(San()\) outputs the real value for individual \(i\). As \(Sam()\) is fixed to output \(i\) independent of its value, there is no difference between \(F_{pre}\) and \(F_{post}\).
    \item In case \(San()\) outputs a randomized answer \(Sam()\) again is not influenced by the outcome of any of the coin tosses and passes \(i\) along to the output. This is of course also independent of the actual randomized result.
  \end{enumerate}
  This concludes the proof that sampling and randomized response are independent regarding their order of execution and thus commute.
\end{proof}

\begin{corollary}
    \emph{(Arbitrary sampling around randomized response)}
  \label{cor:arbitrary-sampling}
  Let \(s = uv\) for \(s,u,v \in (0,1)\) be sampling probabilities for a sampling function \(Sam()\) and \(San()\) be a two-coin randomized response mechanism with probabilities \((p,q)\). Sampling can be arbitrarily traded between pre\-/sampling and post\-/sampling around the randomized response mechanism \(San\).
  \begin{align*}
    San(Sam(D,s),p,q)&\approx Sam(San(Sam(D,u),p,q),v) \\
    &\approx Sam(San(D,p,q),s).
  \end{align*}
\end{corollary}
\begin{proof}
  This follows directly from applying Lemma~\ref{lem:sampling-commutes} and Lemma~\ref{lem:sampling-randomizing-commutes}.
\end{proof}

We will now give a bound on \(\epsilon_{zk}\) for the privacy of our system under the zero\-/knowledge privacy setting, as well as derive a tighter bound for (\(\epsilon_{dp}\))\-/differential privacy, than the bound that generally follows from zero\-/knowledge privacy.

\begin{theorem}
  \emph{(\(\epsilon_{zk}\)\-/zero\-/knowledge privacy)}
  \label{thrm:zkp}
  Let \(A\) be an algorithm that applies sampling with probability \(s\), together with a two-coin randomized response algorithm using probabilities \((p,q)\). \(A\) is \(\epsilon_{zk}\)\-/zero\-/knowledge private with
  \begin{equation}\label{eq:ezk}
    \epsilon_{zk} = ln\left(s\frac{2-s}{1-s}\left(\frac{p+\left(1-p\right)q}{\left(1-p\right)q}\right)+\left(1-s\right)\right).
  \end{equation}
\end{theorem}

The system design is described in Section~\ref{sec:design}.
\begin{proof}
  From~\cite{crowd-blending}, Theorem 1 follows that a \((k,\epsilon_{rr})\)\-/crowd\-/blending private mechanism combined with a pre\-/sampling using probability \(s\) achieves \(\epsilon\)\-/zero\-/knowledge privacy with \[\epsilon_{zk}=ln\left(s\cdot\left(\frac{2-s}{1-s}e^{\epsilon_{rr}}\right)+(1-s)\right).\] We omit the description for the additive error \(\delta\), which can be derived equivalently from~\cite{crowd-blending} Theorem 1. Following Proposition 1 from~\cite{crowd-blending} every \(\epsilon_{rr}\)-differentially private mechanism is also \(k,\epsilon_{rr}\)\-/crowd\-/blending private, thus randomized response being an \(\epsilon_{rr}\)-differentially private mechanism, also satisfies \((k,\epsilon_{rr})\)\-/crowd\-/blending privacy with \(k=1\). Combining both results with \eqref{privacy-level2} \(\epsilon_{rr}=ln\left(\frac{p+(1-p)q}{(1-p)q}\right)\) gives an \[\epsilon_{zk}=ln\left( s\cdot\left(\frac{2-s}{1-s} \left( \frac{p+(1-p)q}{(1-p)q} \right)  \right)+(1-s) \right)\]
zero\-/knowledge private mechanism for randomized response combined with pre\-/sampling. Using Corollary~\ref{cor:arbitrary-sampling} we can replace pre\-/sampling with a combination of pre- and post\-/sampling (with probabilities \(u,v\) respectively and \(s=u \cdot v\)) while keeping \(\epsilon_{zk}\) fixed. We thus have
\[\epsilon_{zk} = ln\left(uv\frac{2-uv}{1-uv}\left(\frac{p+\left(1-p\right)q}{\left(1-p\right)q}\right)+\left(1-uv\right)\right).\]
\end{proof}

If we do not aim at achieving zero\-/knowledge privacy, we can fall back to differential privacy using the result from~\cite{zero-knowledge-privacy}, Proposition 3, which states that any \(\epsilon\)\-/zero\-/knowledge private algorithm is also \(2\epsilon\)-differentially private. Using the results from sampling secrecy~\cite{sampling-privacy2}, which achieve a privacy boost by applying pre\-/sampling before using a differentially private algorithm, we derive a tighter bound for differential privacy, than what follows generally from zero\-/knowledge privacy.

\begin{theorem}
  \emph{(\(\epsilon_{dp}\)\-/differential privacy)}
  \label{thrm:dp}
  Let \(A\) be an algorithm that applies sampling with probability \(s\), followed by a two-coin randomized response algorithm using probabilities \((p,q)\). \(A\) is \(\epsilon_{dp}\)\-/differentially private with
  \begin{equation}\label{eq:edp}
    \epsilon_{dp} = ln\left(1+s\left(\frac{p+\left(1-p\right)q}{\left(1-p\right)q}-1\right)\right).
  \end{equation}
\end{theorem}

\begin{proof}
  We use the result from~\cite{sampling-privacy3}, Proof of Lemma 3, which bounds an \(\epsilon_{rr}\)\-/differential private algorithm combined with pre\-/sampling using probability \(s\) by \(\epsilon_{dp} = ln(1+s(exp(\epsilon_{rr})-1))\). Let \(\epsilon_{rr}=ln\left(\frac{p+\left(1-p\right)q}{\left(1-p\right)q}\right)\) be the bound derived for randomized response, we get
  \[
    \epsilon_{dp}=ln\left(1+s\left(\frac{p+\left(1-p\right)q}{\left(1-p\right)q}-1\right)\right).
  \]
  Applying Corollary~\ref{cor:arbitrary-sampling} we derive an \(\epsilon_{dp}\) bound for the combination of pre\-/sampling, randomized response and post\-/sampling of:
  \[
    \epsilon_{dp}=ln\left(1+(uv)\left(\frac{p+\left(1-p\right)q}{\left(1-p\right)q}-1\right)\right).
  \]
\end{proof}

\begin{proposition}
  \emph{(Tighter \(\epsilon_{dp}\)-differential privacy bound)}
  \label{prop:tighter-dp}
  The bound \(\epsilon_{dp}\) for differential privacy of a sampled randomized response system derived in Theorem~\ref{thrm:dp} is tighter than \(\epsilon_{zk}\)-differential privacy, which is again tighter than the general \(2\epsilon_{zk}\)-differential privacy bound that follows from \(\epsilon_{zk}\)\-/zero\-/knowledge privacy~\cite{zero-knowledge-privacy}.
\end{proposition}

We directly proof Proposition~\ref{prop:tighter-dp} by comparing \(\epsilon_{dp}\) from Theorem~\ref{thrm:dp} with \(\epsilon_{zk}\) from Theorem~\ref{thrm:zkp}. As we want to prove a bound that is tighter than \(\epsilon\), we drop the factor of \(2\). This is possible because a \(\epsilon\)-differentially private algorithm is also \(2\epsilon\)-differentially private. If we succeed in proving a bound tighter than \(\epsilon\), then \(2\epsilon\)-differential privacy is trivially fulfilled.

\begin{proof}
  Proposition~3 from~\cite{zero-knowledge-privacy} states that every \(\epsilon\)\-/zero\-/knowledge private algorithm is also \(2\epsilon\)-differentially private. Using Theorem~\ref{thrm:zkp} we get a \(\epsilon_{zk}\)-differentially private system with \(2\epsilon_{zk}=2 ln\left(s\frac{2-s}{1-s}\left(\frac{p+\left(1-p\right)q}{\left(1-p\right)q}\right)+\left(1-s\right)\right)\). Theorem~\ref{thrm:dp} proves a bound of \(\epsilon_{dp}=ln\left(1+s\left(\frac{p+\left(1-p\right)q}{\left(1-p\right)q}-1\right)\right)\). Let \(e^{\epsilon_{rr}}=\frac{p+\left(1-p\right)q}{\left(1-p\right)q}\). Putting together Theorem~\ref{thrm:dp}, Theorem~\ref{thrm:zkp}, Proposition~\ref{prop:tighter-dp} and Proposition~3~\cite{zero-knowledge-privacy} we have:
  {\scriptsize
  \begin{align*}
    ln\left(1+s\left(\frac{p+\left(1-p\right)q}{\left(1-p\right)q}-1\right)\right) &< ln\left(s\frac{2-s}{1-s}\left(\frac{p+\left(1-p\right)q}{\left(1-p\right)q}\right)+\left(1-s\right)\right) \\
    s\left(\frac{p+\left(1-p\right)q}{\left(1-p\right)q}-1\right) &< \frac{2-s}{1-s} s \left(\frac{p+\left(1-p\right)q}{\left(1-p\right)q}-1\right)\\
    s\left(e^{\epsilon_{rr}}-1\right) &< \frac{2-s}{1-s} s \left(e^{\epsilon_{rr}}-1\right) \\
    1 &< \frac{2-s}{1-s}
  \end{align*}
  }%
  As \(s \in (0,1)\) is the sampling parameter with a minimal right side for \(s_{min}=\argmin_{s\in (0,1)}\left(\frac{2-s}{1-s}\right)=0\) the above inequality becomes \(1<2\), which holds and concludes the proof.
\end{proof}

\begin{figure}[t]
  \centering
  \includegraphics[scale=1.55]{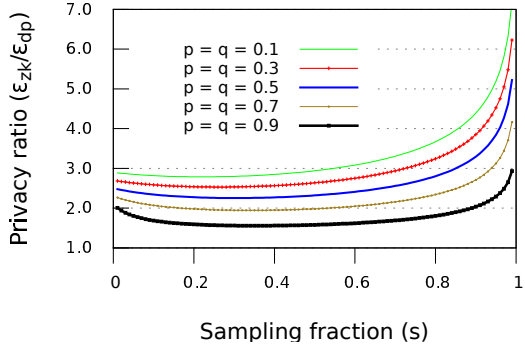} 
  \caption{Ratio of \(\frac{\epsilon_{zk}}{\epsilon_{dp}}\) depending on the sampling parameter \(s\) for different values \(p\) and \(q\).}
  \label{fig:epsilonZKbyDP}
\end{figure}

\myparagraph{Relation of differential privacy and zero\-/knowledge privacy}
Zero\-/knowledge privacy and differential privacy describe the advantage \(\epsilon\) of an adversary in learning information about individual \(i\) by using an output from an algorithm \(San()\) running over database \(D\) containing sensitive information \(a_i \in D\) of individual $i$ compared to using a result of a second --- possibly different --- algorithm \(San^{\prime}()\) running over \(D_{-i}\). Zero\-/knowledge privacy is a strictly stronger privacy metric through the additional access to aggregate information of the remaining database \(D_{-i}\) compared to differential privacy~\cite{zero-knowledge-privacy}. By intuition, as differential privacy is a special case of zero\-/knowledge privacy and the adversary aims at maximizing its advantage, the advantage of an adversary in the zero\-/knowledge model is at least as high and possibly higher than the advantage of an adversary in the differential privacy model: \(\epsilon_{zk} \geq \epsilon_{dp}\). Figure~\ref{fig:epsilonZKbyDP} draws the ratio \(\frac{\epsilon_{zk}}{\epsilon_{dp}}\) between the zero\-/knowledge privacy level \(\epsilon_{zk}\) and the differential privacy level \(\epsilon_{dp}\) given identical parameters \(p,q\) and \(s\). Put differently, as the adversary is allowed to do more in the zero\-/knowledge model, the privacy level is lower, which is reflected by a higher \(\epsilon_{zk}\) value compared to the differential privacy level \(\epsilon_{dp}\) --- given identical system parameters.

\myparagraph{{Property \# II:} Anonymity} We make the following assumptions to achieve the remaining two privacy properties:


\label{sec:assumptions}
\begin{itemize}

\item[{\bf (A1)}] At least two out of the $n$ proxies are not colluding.

\item[{\bf (A2)}] The aggregator does not collude with any of the proxies.

\item[{\bf (A3)}] The aggregator and analysts cannot --- at the same time --- observe the 
communication around the proxies.

\item[{\bf (A4)}] The adversary, seen as an algorithm, lies within the polynomial time complexity class.
\end{itemize}

To provide anonymity, we require that no system component (proxy, aggregator, analyst) can relate a query request or answer to any of the clients. To show the fulfillment of that requirement we take the view of all three parties.

a) A {\em proxy} can of course link the received data stream to a client, as it is directly connected. However, as the data stream is encrypted, it would need to have the plaintext query request or response for the received data stream. To get the plaintext the proxy would either need to break symmetric cryptography, which breaks assumption {\bf (A4)}, collude with {\em all} other proxies for decryption, which breaks assumption {\bf (A1)} or collude with the aggregator to learn the plaintext, which breaks assumption {\bf (A2)}.

b) Anonymity against the {\em aggregator} is achieved by source-re\-writing, which is a standard anonymization technique typically used by proxies and also builds the basis for anonymization schemes~\cite{onion-routing,tor}. To break anonymity the aggregator must be a global, passive attacker, which means that he is able to simultaneously listen to incoming and outgoing traffic of any proxy. This would violate assumption {\bf (A3)}. The other possibility to bridge the proxies is by colluding with any of them --- breaking assumption {\bf (A2)}.

c) The {\em analyst} knows the query request, but doesn't get to learn the single query answers. He needs to collude with the aggregator, to see single responses. Thus the problem reduces to breaking anonymity from the view of the aggregator. Collusion with the aggregator and any proxy would break assumption {\bf (A2)}. Collusion with up to \(n-1\) proxies reduces to breaking anonymity from the proxy view.


\myparagraph{Property \# III:~Unlinkability}
Unlinkability is provided by the source-rewriting scheme as in anonymity. Breaking unlinkability on any {\em proxy} is similar to breaking anonymity, as the proxy would need to get the plaintext query. The {\em aggregator} only gets query results, but no source information, as this is hidden by the anonymization scheme. The query results sent by the clients also do not contain linkable information, just identically structured answers without quasi-identifiers. The view of the {\em analyst} doesn't receive responses, so it must collude with either a proxy or the aggregator, effectively reducing to the same problem as described above.


\myparagraph{Acknowledgements} We would like to thank Amazon for providing us an Amazon Web Services (AWS) Education Grant.

\balance
\referencefontsize
\bibliographystyle{abbrv}
\bibliography{main}

\end{document}